\documentclass[journal,twoside,web]{ieeecolor}

\usepackage{generic}
\usepackage{cite}
\usepackage{amsmath,amssymb,amsfonts}
\usepackage{graphicx}
\usepackage{textcomp}
\usepackage{common_pkgs}
\def\BibTeX{{\rm B\kern-.05em{\sc i\kern-.025em b}\kern-.08em
    T\kern-.1667em\lower.7ex\hbox{E}\kern-.125emX}}
\begin{document}
\title{Quantifying the Value of Preview Information for Safety Control}
\author{Zexiang Liu, \IEEEmembership{Student Member, IEEE}, Necmiye Ozay, \IEEEmembership{Senior Member, IEEE}
\thanks{This research was supported by NSF grant CNS-1931982.
Zexiang Liu and Necmiye Ozay are with the Department of Electrical Engineering and Computer Science at University of Michigan, Ann Arbor, MI 48109 USA (e-mail: \{zexiang, necmiye\}@umich.edu).}
}

\maketitle

\begin{abstract}

Safety-critical systems, such as autonomous vehicles, often incorporate perception modules that can anticipate upcoming disturbances to system dynamics, expecting that such preview information can improve the performance and safety of the system in complex and uncertain environments. 
However, there is a lack of formal analysis of the impact of preview information on safety. In this work, we introduce a notion of safety regret, a properly defined difference between the maximal invariant set of a system with finite preview and that of a system with infinite preview, and show that this quantity corresponding to finite-step preview decays exponentially with preview horizon. Furthermore, algorithms are developed to numerically evaluate the safety regret of the system for different preview horizons. Finally, we demonstrate the established theory and algorithms via multiple examples from different application domains. 
\end{abstract}

\begin{IEEEkeywords}
TBD
\end{IEEEkeywords}

\section{Introduction}
\label{sec:introduction}
The idea of incorporating preview information into controller design has been explored extensively in the past\cite{birla2015optimal}, with many applications to real-world systems, such as autonomous vehicles\cite{liu2019safety, xu2019design, liu2021value}, power systems\cite{sinner2021experimental} and humanoid robots\cite{kajita2003biped}. Compared with purely state or output feedback control mechanisms, preview-based control allows feedforward control based on the available information about the upcoming values of uncertainties affecting the system, and thus can substantially improve the control performance\cite{birla2015optimal}. In the literature, there are several types of uncertainty that are considered as ``previewable". The first type is disturbance or uncontrolled input to the dynamical systems\cite{liu2021value, yu2020power}. Examples include the road curvature for vehicles\cite{xu2019design}, or future wind velocity for wind turbines\cite{sinner2021experimental}, which can be obtained by dedicated perception systems\cite{scholbrock2016lidar, li2020lidar}. The second type of uncertainty is the reference signal in a reference tracking task\cite{tomizuka1975optimal,carrasco2011feedforward}. Different from disturbances, since the reference signal is commonly generated by other algorithmic components, such as a path planner, its future values can be naturally obtained by the tracking controller at run time. 
The third type of previewable uncertainty, recently studied in the context of online optimal control, is the unknown parameters in the optimization problem solved by the controller at run time\cite{li2019online, shin2021controllability}. 
In many cases, these different types of uncertainties (with the corresponding preview information) can be converted to one another. For instance, unknown reference signals are modeled as disturbances in \cite{yu2020power} and as unknown parameters in the cost function in \cite{li2019online}.
Due to this reason,  in this work, we focus on the preview on disturbances in dynamical systems, and
 we demonstrate later how our results can be applied to preview on reference signals with an example.

In this work, we are interested in safety control for systems with preview information. The goal of safety control is to synthesize controllers that can guarantee that the closed-loop system satisfies given safety requirements indefinitely, robust to some degree of uncertainty in the dynamics\cite{bertsekas1972infinite, aubin1990survey, liu2019safety, liu2021value, liu2020scalable}.  In the context of model predictive control (MPC), this is related to the notion recursive feasibility of the online optimization problem solved at each time step so that constraints are satisfied indefinitely \cite{marruedo2002stability, mayne2014model}.
For systems with preview, safe controllers are allowed to use feedback on the states and feedforward on the preview information. 
An important question is if a safe controller should utilize all the available preview information. 
In theory, the more preview information a controller utilizes, the safer the system may become.  
This is verified by numerical examples in our previous works\cite{liu2019safety, liu2021value}, where incorporating more preview information enlarges (i) the region of states where the safety requirements can be enforced or (ii) the range of disturbances the system can tolerate under the safety constraints.   
However, in practice, it is often computationally intractable to incorporate all the available preview information since many safety control algorithms (even the most scalable ones) suffer from the curse of dimensionality\cite{wintenberg2020implicit, herbert2021scalable, brunke2022safe, gurriet2019scalable}.
In those cases, one needs to carefully select the amount of preview information fed into the safe controller, for a good balance between the computational cost and the degradation in safety due to omitting part of the preview information. 

Motivated by this need, in this work, we study how the safety of a system is impacted by the amount of preview. 
Since we focus on the preview of disturbances, in our analysis, the amount of preview information is dictated by the number of time steps that the future disturbances can be previewed, which we call the \emph{preview horizon}.
We measure the safety of the same system with different preview horizons by a notion called \emph{safety regret}, defined based on the maximal robust control invariant set (RCIS) of an auxiliary system that augments the original states with the preview information. Given a preview horizon $p$, the corresponding safety regret reflects the room to improve the system safety as we increase the preview horizon from $p$ to infinity. The main contributions of this work include: 
\begin{itemize}
	\item We provide novel outer approximations of the maximal RCIS for nonlinear systems augmented with preview information, by exploring a duality between preview and input delay (Section~\ref{sec:struct}).
	\item For linear systems, we prove that the safety regret of a finite-step preview decays exponentially fast with the preview horizon. For polytopic state-input constraints, we further develop algorithms that compute upper bounds of the safety regret (Section~\ref{sec:convergence}).
	\item  {We extend our analysis to show how the preview horizon affects the feasible domain of preview-based model predictive controllers under different recursive feasibility constraints} (Section~\ref{sec:mpc}).
	\item We demonstrate the usage of our theoretical results and the proposed algorithms with both analytical and numerical examples (Section~\ref{sec:ex}).
\end{itemize}

\subsection{Related Works}

The value of preview information has been studied under different control frameworks. In reactive synthesis, \cite{kupferman2011synthesis, holtmann2010degrees, klein2015much} study the necessary conditions on the preview horizon under which there exists a controller that realizes a Linear Temporal Logic (LTL) specification. However, these results are only for finite-state transition systems, not for systems with continuous state space studied in this work. In the optimization-based control framework,  \cite{tomizuka1975optimal, shin2021controllability, li2019online, yu2020power} show that the sensitivity of the optimal solution with respect to the preview information \cite{tomizuka1975optimal, shin2021controllability} or the dynamic regret of the optimal controllers parameterized by the preview information \cite{li2019online, yu2020power} decay exponentially with the preview horizon. Since the class of cost functions considered by \cite{tomizuka1975optimal, shin2021controllability, li2019online, yu2020power} cannot characterize state constraints, those results cannot be applied to safety control problems.
For safety control, our previous work \cite{liu2021value} studies the structural properties of RCISs for systems augmented with preview information, which forms a foundation for the theory developed in this work. Its connection to this work will be explained in detail in Section \ref{sec:struct}.

\textbf{Notation:} For vectors $x_1\in \R^{n}$ and $x_2\in \R^{m}$, $(x_1,x_2)\in\R^{n+m}$ denotes the concatenation of them.  For a sequence of indexed vectors $(d_{i})_{i=1}^{p}$, we denote the concatenation $(d_1,d_2, \cdots, d_{p})$ by $d_{1:p}$ for short. For a polytope $\mathcal{P}= \{x \mid Hx\leq h\} $, the $H$-representation of $\mathcal{P}$ is the matrix $[H\ h]$. Given a non-negative scalar $ \lambda \geq 0$ and a set $X$, $ \lambda X = \{ \lambda x \mid  x\in X\}$. The projection of a set  $X$ in $ \R^{n+m}$ onto the first $n$ coordinates is denoted by ${\Proj{n}(X) = \{x_1\in \R^{n} \mid  (x_1,x_2)\in X\} }$. {For two subsets $X$ and $Y$ of $\R^n$, their sum is defined by $X+Y = \{x+y \mid x\in X, y\in Y\}$.} For a positive integer $q$, the set $\{1,2, \cdots, q\} $ is denoted by $[q]$. The unit hypercube in $ \R^{n}$ is denoted by $ \mathcal{B}(n) = [-1,1]^{n}.$ The Hausdorff distance between two sets $X$ and $Y$ in $\R^{n}$ is induced from $2$-norm in $\R^{n}$ and is denoted by $d(X,Y)$. Given a compact set $X$ in $\R^{n}$, the radius of $X$ with respect to a point $x_0$ is defined by $\psi(X) = \sup_{x\in X} \Vert x -x_0\Vert_{2}$.

\section{Preliminaries} \label{sec:prelim} 
\subsection{Discrete-time systems with preview} \label{sec:def} 
A discrete-time dynamical system $ \Sigma $ is in form
\begin{align} \label{eqn:sys} 
\Sigma: x(t+1) = f(x(t),u(t),d(t)),
\end{align}
with state $ x\in \R^n $, input $u\in \R^{m}$ and disturbance ${d\in D \subseteq \R^{l}}$. The disturbance set $D$ is assumed to be compact. A system $\Sigma$ is said to have \emph{$p$-step preview} if at each time step  $t$, the controller has access to the measurements of
\begin{itemize}
	\item the current state $x(t)$, and
	\item the current and incoming disturbances $\{d(t+k)\}_{k=0}^{p-1}$ in $p$ steps, denoted by $d_{1:p}(t)$ for short.
\end{itemize}
For a system $\Sigma$ with $p$-step preview, we define an augmented system whose states stack the states $x(t)$ and the previewed disturbances $d_{1:p}(t)$ of $\Sigma$ together, called the $p$-augmented system $\Sigma_{p}$ of $\Sigma$. The $p$-augmented system $\Sigma_{p}$ has the dynamics
\begin{align} \label{eqn:sys_p} 
\Sigma_{p}: \begin{bmatrix}
x(t+1)\\d_1(t+1)\\\vdots \\ d_{p-1}(t+1)	\\d_{p}(t+1)
\end{bmatrix}  = \begin{bmatrix}
f(x(t),u(t),d_1(t))\\
d_2(t)\\ 
\vdots\\
d_{p}(t)\\
d(t) 
\end{bmatrix}
\end{align}
with state $(x,d_1, \cdots, d_{p})\in \R^{n+pl}$, input $u\in \R^{m}$ and disturbance $d\in D \subseteq \R^{l}$. 
Any preview-based controller of $\Sigma$ with feedback on the state $x(t)$  and feedforward on the previewed disturbances $d_{1:p}(t)$ is equivalent to a state-feedback controller of the $p$-augmented system $\Sigma_{p}$. 
Due to this equivalence relation, we study the safety control for the system $\Sigma$ with $p$-step preview by simply studying the state-feedback safety control for its $p$-augmented system $\Sigma_{p}$. 

\subsection{Robust controlled invariant sets} \label{sec:inv} 
In this subsection, we first define RCISs for discrete-time systems in the form \eqref{eqn:sys}, and then focus on the RCISs of the $p$-augmented system $\Sigma_{p}$. 

Let $S_{xu} \subseteq \R^{n+m}$ be the \emph{safe set} of $\Sigma$, namely the set of allowable state-input pairs. 
Given an initial state $x_0$, we call a state-feedback controller ${\cont_{safe}: \R^{n} \rightarrow \R^{m}}$ a \emph{safe controller} if any state-input trajectory of the system  $\Sigma$ under the control of $\cont_{safe}$ with the initial state $x_0$ stays within the safe set $S_{xu}$ indefinitely. 
\begin{definition} \label{def:rcis} 
	A set $C \subseteq \R^{n}$ is a \emph{robust controlled invariant set} (RCIS) of the system $\Sigma$ in safe set $S_{xu} \subseteq \R^{n+m}$ if for all $x\in C$, there exists some $u\in \R^{m}$ such that $(x,u)\in S_{xu}$ and for all $d\in D$, $f(x,u,d)\in C$.  An RCIS $C_{max}$ is the \emph{maximal RCIS} of $\Sigma$ in $S_{xu}$ if $C_{max}$ contains all the RCISs of $\Sigma$ in $S_{xu}$.
\end{definition}

The maximal RCIS $C_{max}$ is {well-defined due to} the fact that the union of RCISs is still an RCIS.
According to Definition \ref{def:rcis}, it can be shown that the maximal RCIS in $S_{xu}$ characterizes all the safe controllers of the system in the following sense:

(i)  the maximal RCIS $C_{max}$ of $\Sigma$ in $S_{xu}$ contains all the initial states $x_0$ where a safe controller exists;  

(ii) for any initial state $x_0\in C_{max}$, a state-feedback controller $\cont:\R^{n} \rightarrow \R^{m}	$ is a safe controller if and only if $(x,\cont(x))\in S_{xu}$ and $f(x,\cont(x),D) \subseteq C_{max}$ for all states $x$ reachable from $x_0$ {with the controller} $\cont$.  

Due to the importance of the maximal RCIS in safety control, in this work, we analyze the impact of preview on the safety of the system $\Sigma$ by studying how the maximal RCIS varies with different preview horizon. More specifically, given the safe set $S_{xu}$ of $\Sigma$, we define the safe set $S_{xu,p}$ of the $p$-augmented system $\Sigma_{p}$ of $\Sigma$ by 
\begin{align} \label{eqn:safe_set_p} 
S_{xu,p}= \{ (x,d_{1:p},u) \mid (x,u)\in S_{xu}, d_{1:p}\in D^{p}\}.
\end{align}
Intuitively, the safe set $S_{xu,p}$ imposes the same constraints on $(x,u)$ as in $S_{xu}$ and does not constrain the previewed disturbances $d_{1:p}$ since we cannot control the disturbance inputs. 

We denote the maximal RCIS of the system $\Sigma_{p}$ in the safe set $S_{xu,p}$ by $C_{max,p}$. The shape and dimension of $C_{max,p}$ vary with the preview horizon $p$.
Since $C_{max,p}$ characterizes all the safe controllers of $\Sigma_{p}$, the variation of $C_{max,p}$ with the preview horizon $p$ reflects the impact of different preview horizons on the safety of $\Sigma$. 
In this work, we analyze how $C_{max,p}$  varies with different preview horizons $p$, and its implication to the safety of the system. 
Note that it is intractable to compute $C_{max,p}$ for all $p$ and then compare their difference. 
Later we introduce a quantity called \emph{safety regret} to evaluate the variation of $C_{max,p}$ with the preview horizon $p$. We show that it is possible to estimate the safety regret without computing $C_{max,p}$, by exploiting the structure of the augmented system $\Sigma_{p}$ and its safe set $S_{xu,p}$. 

\subsection{Necessary definitions}
The following definitions are essential for the proceeding theoretical and algorithmic discussion.
\begin{definition}[{$Pre_{SF}(X)$ in \cite{blanchini2008set}}]  \label{def:pre} 
	Given a set $X \subseteq \R^{n}$, the \emph{one-step backward reachable set} $Pre_{\Sigma}(X, S_{xu})$ of $X$ for the system $\Sigma$ constrained in the safe set $S_{xu}$ is defined as
	\begin{align}
	\begin{split}
	&Pre_{\Sigma}(X,S_{xu}) =\\
	&\{x \mid \exists u,  (x,u)\in S_{xu}, f(x,u,d) \in X, \forall d\in D\}. 
	\end{split}
	\end{align}
\end{definition}
 Given Definition \ref{def:pre}, the \emph{$k$-step backward reachable set} ${Pre_{\Sigma}^{k}(X,S_{xu})}$ is defined recursively by
\begin{align}
Pre_{\Sigma}^{1}(X,S_{xu}) &= Pre_{\Sigma}(X,S_{xu}),\\
Pre_{\Sigma}^{k}(X,S_{xu})&= Pre_{\Sigma}^{k-1}(X,S_{xu}) \label{eqn:pre_k} .
\end{align}

The definitions of RCIS and $k$-step backward reachable set are also valid for systems without disturbances, that is the disturbance set $D=\{0\}$. In that case, an RCIS is called a \emph{controlled invariant set  (CIS)}. In addition, for systems without disturbance, we have the following definition.  
\begin{definition}\label{def:N-step-lambda}
	A set $X$ is an \emph{$N$-step $ \lambda$-contractive set} of the system $\Sigma$ in the safe set $S_{xu}$ if $X$ satisfies
	\begin{align}
X \subseteq	Pre_{\Sigma}^{N} ( \lambda X, S_{xu}).
	\end{align}
\end{definition}
 {For $N=1$, we call $X$ a $\lambda$-contractive set by convention\cite{blanchini1994ultimate,blanchini2008set}.} 
Intuitively, $X$ is $N$-step $ \lambda$-contractive if we can steer the system from any state in $X$ to a state in $ \lambda X$ in $N$ steps while keeping the state-input trajectory in the safe set. Note that an $N$-step $ \lambda$-contractive set is not necessarily a CIS. 

\section{The structure of the maximal RCIS for systems with preview} \label{sec:struct} 
In this section, we present several structural properties of the maximal RCIS $C_{max,p}$ of the  $p$-augmented system, which allow us to approximate $C_{max,p}$ without computing it, and pave the way to our analysis in the following sections.

\subsection{Inner and outer bounds of the maximal RCIS $C_{max,p}$} \label{sec:bnds} 
We briefly summarize the inner and outer bounds of the maximal RCIS $C_{max,p}$ derived in the previous work \cite{liu2021value}.  Those bounds are useful when the actual representation of $C_{max,p}$ is difficult to obtain, and will be crucial in our analysis of the variation of $C_{max,p}$ with different $p$. 
 
To have an outer bound of the maximal RCIS $C_{max,p}$, we introduce an auxiliary system: Given a system $\Sigma$ and a safe set $S_{xu}$, we define the \emph{disturbance collaborative system} $\mathcal{D}(\Sigma)$ of $\Sigma$ by
\begin{align} \label{eqn:sys_dis} 
   \mathcal{D}(\Sigma): x(t+1) = f(x(t),u(t),u_{d}(t)),
\end{align}
with state $x\in \R^{n}$ and two inputs $u\in \R^{m}$, $u_{d}\in \R^{l}$.
The safe set $S_{xu,co}$ of $\mathcal{D}(\Sigma)$ is $S_{xu}\times D$, that is 
\begin{align}
    S_{xu,co} = \{(x,u,u_{d}) |(x,u)\in S_{xu},u_{d}\in D \}. 
\end{align}
We denote the maximal CIS\footnote{The set $C_{max,co}$ is a CIS instead of an RCIS since $\mathcal{D}(\Sigma)$ has no disturbance.} of $\mathcal{D}(\Sigma)$ in $S_{xu,co}$ by $C_{max,co}$.  
The only difference between the original system $\Sigma$ and the corresponding $\mathcal{D}(\Sigma)$ is that we turn the disturbance term $d$ in $\Sigma$ into the second control input $u_{d}$ in $\mathcal{D}(\Sigma)$. Due to the extra control authority introduced by $u_{d}$, the maximal CIS $C_{max,co}$ of $\mathcal{D}(\Sigma)$ provides an outer bound of the maximal RCIS $C_{max,p}$ for all $p$. This outer bound along with an inner bound of $C_{max,p}$ are stated in the following theorem.  

\begin{theorem}[\cite{liu2021value}, Theorems 1 and 2]  \label{thm:bnds} 
For a system $\Sigma$ with $p$-step preview, the maximal RCIS $C_{max,p}$ of $\Sigma_{p}$ in $S_{}$ is inner approximated by $C_{max,p'}\times D^{p-p'}$ for any $p'<p$, and is outer approximated by the Cartesian product $C_{max,co}\times D^{p}$. That is,
\begin{align} \label{eqn:bnds} 
C_{max,p'}\times D^{p-p'} \subseteq C_{max,p} \subseteq C_{max,co}\times D^{p}.
\end{align}
Furthermore,  the inner bound $C_{max,p'}\times D^{p-p'}$ is an RCIS of $\Sigma_{p}$ in $S_{xu,p}$.
\end{theorem}
A formal proof of Theorem \ref{thm:bnds} can be found in \cite{liu2021value}. As discussed, the outer bound is thanks to the extra control authority.
The inner bound is based on the intuition that a longer preview horizon $p$ should make the maximal RCIS $C_{max,p}$ larger. However, this intuition is not fully correct. Since the dimensionality of $\Sigma_{p}$ grows with $p$, the maximal RCIS $C_{max,p}$ lies in a different space for each different $p$ and cannot be compared with each other directly. It turns out that for a shorter preview horizon $p' < p$, we can always lift an RCIS  $C_{p'}$ of $\Sigma_{p'}$ to the RCIS $C_{p'}\times D^{p-p'}$  of $\Sigma_{p}$ and then compare the lifted set with $C_{max,p}$. The difference between the lifted RCIS $C_{max,p'}\times D^{p-p'}$ and the maximal RCIS $C_{max,p}$ tells how much we gain in safety by increasing the preview horizon from $p'$ to $p$.

\subsection{Improved outer bounds for the maximal RCIS $C_{max,p}$}
The outer bound of $C_{max,p}$ in Theorem \ref{thm:bnds} is based on the maximal CIS $C_{max,co}$ of $\mathcal{D}(\Sigma)$ of $\Sigma$. 
Since $C_{max,co}$ is independent of the preview horizon $p$, this outer bound is not necessarily tight. 
 In this subsection, we derive tighter outer bounds of $C_{max,p}$, by exploring a duality between delay and preview.

Apart from the disturbance collaborative system $\mathcal{D}(\Sigma)$, we can also define the disturbance collaborative system $\mathcal{D}(\Sigma_{p})$ of the $p$-augmented system $\Sigma_{p}$.
Formally, the system $\mathcal{D}(\Sigma_{p})$ is in form
\begin{align}
	\mathcal{D}(\Sigma_{p}):	\begin{bmatrix}
		x(t+1)\\
		d_1(t+1)\\
		\vdots\\
		d_{p-1}(t+1)\\
		d_{p}(t+1)
	\end{bmatrix} = 
	\begin{bmatrix}
		f(x(t),u(t),d_{1}(t))\\
		d_2(t)\\
		\vdots\\
		d_{p}(t)\\
		u_{d}(t)
\end{bmatrix} \end{align} 
with control inputs $u\in \R^{n}$ and $u_{d}\in \R^{l}$. 
Let the safe set of the system $\mathcal{D}(\Sigma_{p})$ be ${S_{xu,p}\times D}$. We denote the maximal CIS of the system $\mathcal{D}(\Sigma_{p})$ in safe set $S_{xu,p}\times D$ by $C_{max,p,co}$.  When $p=0$, $C_{max,0,co}= C_{max,co}$.

Then, for any $p'\leq p$, the system $\Sigma_{p}$ can be viewed as the $(p-p')$-augmented system of the system $\Sigma_{p'}$. By applying Theorem \ref{thm:bnds} to $\Sigma_{p'}$, we have the following corollary. 
\begin{corollary} \label{cor:tight_outer_bnd} 
    For a system $\Sigma$ with $p$-step preview and any non-negative integer $p'\leq p$, the maximal RCIS $C_{max,p}$ is outer approximated by the Cartesian product $C_{max,p',co}\times D^{p-p'}$. That is,
	\begin{equation} \label{eqn:tight_outer_bnd} 
	    C_{max,p} \subseteq C_{max,p',co}\times D^{p-p'}.
	\end{equation}
\end{corollary}
According to Corollary \ref{cor:tight_outer_bnd}, we have $(p+1)$ outer bounds for the maximal RCIS $C_{max,p}$, that is ${\{C_{max,k,co}\times D^{p-k}\}_{k=0}^{p}}$, including the one in Theorem \ref{thm:bnds}. 

Next, we want to study the relation between those outer bounds, and figure out which one is the tightest bound.   
The key is to realize that  $\mathcal{D}(\Sigma_{p})$ is actually the state-space representation of  $\mathcal{D}(\Sigma)$ with $p$-step input delay in $u_{d}$: 
\begin{align} \label{eqn:sys_delay} 
 x(t+1)= f(x(t),u(t),u_{d}(t-p)).
\end{align}
In \cite{liu2020scalable}, $\mathcal{D}(\Sigma_{p})$ is called the augmented system of the input delay system in  \eqref{eqn:sys_delay}. 
Intuitively, since we turn the disturbance in $\Sigma_{p}$ to the input $u_{d}$ in $\mathcal{D}(\Sigma_{p})$, the $p$-step preview on $d$ becomes a $p$-step delay on the input $u_{d}$. This is what we call the duality between delay and preview. The relations of the four systems $\Sigma$, $\Sigma_{p}$, $\mathcal{D}(\Sigma)$ and $\mathcal{D}(\Sigma_{p})$ are shown by the diagram in Fig. \ref{fig:delay_preview}. Since $\mathcal{D}(\Sigma_{p})$ is the ``delayed" version of  $\mathcal{D}(\Sigma)$, the maximal CIS of $\mathcal{D}(\Sigma_{p})$ is embedded in the maximal CIS of $\mathcal{D}(\Sigma)$, shown by the next theorem.  

\begin{figure}[]
	\centering
	\includegraphics[width=0.4\textwidth]{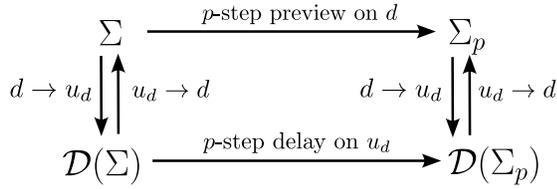}
	\caption{The relation diagram of the four systems $\Sigma$, $\mathcal{D}(\Sigma)$, $\Sigma_{p}$, and $\mathcal{D}(\Sigma_{p})$.}
	\label{fig:delay_preview}
\end{figure}
 
\begin{theorem} \label{thm:CIS_delay} 
For any $p\geq 0$, the maximal CIS $C_{max,p,co}$ of the system $\mathcal{D}(\Sigma_{p})$ in the safe set $S_{xu,p}\times D$ can be obtained from the maximal CIS $C_{max,co}$ of $\mathcal{D}(\Sigma)$ in $S_{xu}\times D$, via the formula
\begin{align}  
     \label{eqn:C_max_p_co} 
		\begin{split}
			C_{max,p,co}=& {Pre_{\mathcal{D}(\Sigma_p)}^p (C_{max,co}\times D^p)}
		\end{split}
\end{align}
where $x(t) = f(x(t-1),u(t-1), d_{t}(0))$ for $0\leq t\leq p$.
Furthermore, we have
	\begin{align} \label{eqn:delay_relation} 
	  C_{max,p,co} \subseteq  C_{max,co}\times D^{p}.
	\end{align}
\end{theorem}
Theorem \ref{thm:CIS_delay} extends the results in \cite{liu2020scalable} and actually holds for any deterministic system with input delays.   
According to \eqref{eqn:delay_relation}, it can be shown that  for any $p'\leq p$,
 \begin{align}
    C_{max,p,co} \subseteq C_{max,p',co} \times D^{p-p'} \subseteq C_{max,co}\times D^{p}.
 \end{align}
Thus, $C_{max,p,co}$ is the tightest outer bound of $C_{max,p}$ in the set ${\{C_{max,k,co}\times D^{p-k}\}_{k=0}^{p}}$.  Furthermore, this outer bound can be computed by \eqref{eqn:C_max_p_co},  once the maximal RCIS $C_{max,co}$ is known. Given $C_{max,co}$, the formula in \eqref{eqn:C_max_p_co} amounts to computing $Pre$ only $p$ times, whose cost is typically negligible compared with the computation cost of $C_{max,co}$, which involves recursive computations of $Pre$ until convergence. Thus, we improve the outer bound of $C_{max,p}$ with little extra cost.  

Due to the similarity between the system pairs $(\Sigma, \Sigma_{p})$ and $(\mathcal{D}(\Sigma), \mathcal{D}(\Sigma_{p}))$,  one may wonder if the maximal RCIS $C_{max,p}$ of the $p$-augmented system can also be obtained from $C_{max}$ by a formula similar to \eqref{eqn:C_max_p_co}. Unfortunately, we cannot find such a formula. But the duality between the preview and delay, as shown in Fig. \ref{fig:delay_preview}, does allow us to take the advantage of \eqref{eqn:C_max_p_co} to obtain the tighter outer approximations of $C_{max,p}$ easily.   

\section{Quantifying the value of preview} \label{sec:convergence} 
In this section, we show how the value of preview information in safety control varies with the preview horizon $p$. First, we need to find a way to quantify the value of different preview horizons. Ideally, we want to quantify the value of preview by the size of the maximal RCIS $C_{max,p}$, since this set characterizes all the safe controllers. However, since the dimension of  $C_{max,p}$ depends linearly on $p$, it is not possible to compare the size of $C_{max,p}$ over different preview horizon $p$. To resolve this issue, we project $C_{max,p}$ onto the first $n$ coordinates, and study how the size of the projections $\Proj{n}(C_{max,p})$ varies with $p$. 

Here is why the size of the projection $\Proj{n}(C_{max,p})$ indeed reflects the value of $p$-step preview. First,  the projection $\Proj{n}(C_{max,p})$ contains all the initial states $x_0$ of the system $\Sigma$ where a safe controller exists for some initial $p$-step preview. 
Second, by Theorem \ref{thm:bnds}, for any $p'\leq p$, 
\begin{align}
    \Proj{n}(C_{max,p'}) \subseteq \Proj{n}(C_{max,p}).
\end{align}
That is, the projection $\Proj{n}(C_{max,p})$ expands with $p$, which matches our intuition that a longer preview horizon has higher value in safety control.

In the remainder of this section, we show what the limit of the projection $\Proj{n}(C_{max,p})$ is and how fast this projection converges to the limit in Hausdorff distance. For short, the Hausdorff distance between the projection $\Proj{n}(C_{max,p})$ and its limit is denoted by $d_{p}$, that is
\begin{align} \label{eqn:dp} 
    d_p = d(\Proj{n}(C_{max,p}), \lim_{p \to \infty} \Proj{n}(C_{max,p})).
\end{align}
Intuitively, the value of $d_{p}$ reflects the room to improve the safety of the system if we are allowed to further increase the preview horizon. In a sense, $d_{p}$ measures the safety gap between the $p$-step preview and the infinite-step preview. Due to this reason, we also call $d_{p}$ the \emph{safety regret} of the $p$-step preview (similar to the notion of regret in \cite{yu2020power, li2019online}).

\subsection{Assumptions} \label{sec:asp} 
We restrict our analysis to the class of discrete-time linear systems. A system $\Sigma$ is \emph{linear} if its transition function $f$ in \eqref{eqn:sys}  is in form
\begin{align} \label{eqn:sys_linear} 
  f(x,u,d) = Ax + Bu +Ed,
\end{align}
with matrices $A\in \R^{n\times n}$, $B\in \R^{n\times m}$ and $E\in \R^{n\times l}$. The results later in this section are based on the following assumptions. 

\begin{assumption} \label{asp:asymp_stab} 
    The disturbance collaborative system $ \mathcal{D}(\Sigma)$ of the linear system $\Sigma$ is stabilizable. 
\end{assumption}   

Note that the disturbance-collaborative system $\mathcal{D}(\Sigma)$ being stabilizable is a weaker condition than the system $\Sigma$ being stabilizable, since the system $\mathcal{D}(\Sigma)$ has one more control input $u_{d}$ than the system $\Sigma$.

\begin{assumption} \label{asp:cmpt} 
    The safe set $S_{xu}$ and the disturbance set $D$ are convex and compact.  
\end{assumption}

\begin{lemma} \label{lem:forced_eq} 
Suppose that the system $\Sigma$ is linear.   Under Assumption \ref{asp:cmpt}, if the set $\Proj{n}(C_{max,p})$ is nonempty, then 

(i) $\Proj{n}(C_{max,p})$ is a convex compact CIS of the system $\mathcal{D}(\Sigma)$ within $S_{xu}\times D$; 

(ii) there exists a forced equilibrium $(x_{e},u_{e},d_{e})\in S_{xu}\times D$ of the system $\mathcal{D}(\Sigma)$ such that $x_{e}$ is in $\Proj{n}(C_{max,p})$.
\end{lemma}

\begin{assumption} \label{asp:int} 
	For some $p_0\geq 0$, there exists a forced equilibrium $(x_{e},u_{e},d_{e})$ of $\mathcal{D}(\Sigma)$ in the interior of ${S_{xu}\times D}$ with $x_{e}\in \Proj{n}(C_{max,p_0})$.
\end{assumption}
According to Lemma \ref{lem:forced_eq}, Assumption \ref{asp:int} is almost an implication of Assumption \ref{asp:cmpt}, except that we require the forced equilibrium is not only in the safe set $S_{xu}\times D$, but in its interior. Therefore, Assumption \ref{asp:int} is not that restrictive.

\begin{remark} \label{rem:origin} 
For linear systems, we can shift the origin of the state space to any forced equilibrium without changing the system equations. Hence without loss of generality, for the remainder of this section,  we simply assume that the forced equilibrium $(x_{e},u_{e},d_{e})$ in Assumption \ref{asp:int} is the origin of the state-input-disturbance space $\R^{n+m+l}$. 
\end{remark}

\subsection{Convergence of $\Proj{n}(C_{max,p})$} \label{sec:conv_asym} 
By Remark \ref{rem:origin}, the safe set $S_{xu}\times D$, the maximal RCIS $C_{max,co}$ and the projection $\Proj{n}(C_{max,p})$ all contain the origin for any $p\geq p_0$. Thus, there exists a scalar $ \lambda \in (0,1]$ such that
\begin{align} \label{eqn:lambda_0} 
	0 \in \lambda C_{max,co} &\subseteq \Proj{n}(C_{max,p_0}) \subseteq C_{max,co}.
\end{align}
{We call the maximal $ \lambda$ such that \eqref{eqn:lambda_0} holds the \emph{initial factor} $\lambda_0$, which reflects the portion of the set $C_{max,co}$ contained in the projection of the set $C_{max,p_{0}}$.} By definition, the initial factor $ \lambda_0 \geq 0$. To prove the convergence of $\Proj{n}(C_{max,p})$, we need the initial factor $ \lambda_0 > 0$, which is ensured by the following assumption. 
\begin{assumption}\label{asp:int2}
The forced equilibrium $(x_{e},u_{e}, d_{e})$ in Assumption \ref{asp:int} satisfies that $x_{e}$ is in the interior of $\Proj{n}(C_{max,p_0})$.
\end{assumption}

 By Lemma \ref{lem:forced_eq}, whenever the maximal RCIS $C_{max,p}$ is nonempty, Assumption \ref{asp:cmpt} implies that the subspace $\mathcal{X}_{eq}\subseteq \R^{n+m+l}$  of all the forced equilibria intersects with the convex set $\mathcal{C}:= (S_{xu}\times D)\cap(\Proj{n}(C_{max,p})\times \R^{m+l})$. Assumptions \ref{asp:int} and \ref{asp:int2} hold if and only if  $\mathcal{X}_{eq}$ intersects with the interior of  $\mathcal{C}$, hence they only slightly stricten Assumption \ref{asp:cmpt}.
Numerically Assumptions \ref{asp:int} and \ref{asp:int2} can be verified by a linear program, provided in the next subsection.
\begin{lemma} \label{lem:pre_expand} 
For any system $\Sigma$ in \eqref{eqn:sys} with safe set $S_{xu}$ and a preview horizon $p$, the projection $\Proj{n}(C_{max,p})$ satisfies that for any $k>0$,
\begin{align} \label{eqn:k_step_contain} 
	&Pre_{\mathcal{D}(\Sigma)}^{k}(\Proj{n}(C_{max,p}), S_{xu}\times D)\nonumber\\
	&\subseteq \Proj{n}(C_{max,p+k}) \subseteq C_{max,co}.
\end{align}
\end{lemma}

By \eqref{eqn:lambda_0} and Lemma \ref{lem:pre_expand}, we obtain the following inner bound of $\Proj{n}(C_{max,p_0+k})$:
\begin{align} \label{eqn:inner_bound_proj} 
	&Pre_{\mathcal{D}(\Sigma)}^{k}( \lambda_0 C_{max,co}, S_{xu}\times D)\nonumber\\
	&\subseteq \Proj{n}(C_{max,p_0+k}) \subseteq C_{max,co}.
\end{align}
Since $ \lambda_0 C_{max,co}$ is a CIS of $\mathcal{D}(\Sigma)$ in $S_{xu}\times D$, {the $k$-step backward reachable set of $\lambda_0 C_{max,co}$ in $S_{xu}\times D$} is non-shrinking with $k$. If we can show that  this $k$-step backward reachable set converges to the maximal CIS $C_{max,co}$ of $\mathcal{D}(\Sigma)$ in $S_{xu}\times D$ as $k$ goes infinity, then by \eqref{eqn:inner_bound_proj}, the projection $\Proj{n}(C_{max,p})$ converges to the maximal CIS $C_{max,co}$. Furthermore, if we know how fast $Pre_{\mathcal{D}(\Sigma)}^{k}( \lambda_0 C_{max,co}, S_{xu}\times D)$ converges, then we have a lower bound on the convergence rate of $\Proj{n}(C_{max,p})$. The following theorem gives such a lower bound{, inspired by the contraction analysis of set-valued mappings in  \cite{artstein2008feedback, darup2016computation, liu2022convergence}.}

\begin{theorem} \label{thm:conv_suff} 
For a linear system $\Sigma$ with a safe set $S_{xu}$ satisfying Assumption \ref{asp:cmpt}, suppose that there exists a scalar $ {\gamma\in (0,1]}$, a positive integer $N$ and a scalar $ \lambda\in [0,1]$ such that $ \gamma C_{max,co}$ is an $N$-step $ \lambda$-contractive CIS of the system $\mathcal{D}(\Sigma)$ within the safe set $S_{xu}$. Then, the $kN$-step backward reachable set of $ \lambda_0 C_{max,co}$ satisfies that for $k \leq  k_0$,
\begin{align}
	\frac{ \lambda_0}{ \lambda^{k}} C_{max,co} \subseteq Pre_{\mathcal{D}(\Sigma)}^{kN}( \lambda_0 C_{max,co}, S_{xu}\times D);  \label{eqn:C_max_conv_1}
\end{align}
for $k\geq k_0$,
\begin{align}
	&\left( 1 - (1-  \lambda_0/ \lambda^{k_0}) \left( \frac{1- \gamma}{ 1 - \gamma \lambda}  \right)^{k-k_0} \right) C_{max,co} \nonumber \\
	& \subseteq Pre_{\mathcal{D}(\Sigma)}^{kN}( \lambda_0 C_{max,co}, S_{xu}\times D), \label{eqn:C_max_conv_2}
\end{align}
where  
\begin{align}
k_{0} &=  \max\left(0, \bigg\lceil  \frac{ \log \lambda_0 - \log \gamma}{ \log \lambda} -1 \bigg\rceil  \right). \label{eqn:k_0}
\end{align}
\end{theorem}
It turns out that for stabilizable systems satisfying assumptions in Section \ref{sec:asp}, the $N$-step $ \lambda$-contractive CIS stated in Theorem \ref{thm:conv_suff} always exists, shown by the following lemma. 
\begin{lemma} \label{lem:contraction} 
	Under Assumptions \ref{asp:asymp_stab}, \ref{asp:cmpt} and \ref{asp:int}, there exist a scalar $ \gamma \in (0,1]$, a positive integer $N$ and a scalar $ \lambda\in [0,1)$ such that $\gamma C_{max,co}$ is an $N$-step $ \lambda$-contractive CIS of the system  $\mathcal{D}(\Sigma)$ within the safe set $S_{xu}$, that is
	\begin{align}
	 \gamma C_{max,co} \subseteq
   Pre_{\mathcal{D}(\Sigma)}^{N}( \lambda \gamma C_{max,co},S_{xu}).
	\end{align}
\end{lemma}

{Intuitively, Assumptions \ref{asp:asymp_stab} and \ref{asp:int} ensure that the maximal CIS $C_{max,co}$ of $\mathcal{D}(\Sigma)$ contains a $\lambda'$-contractive ellipsoid $\mathcal{E}$ centered at the origin with $\lambda'<1$ \cite{de2004computation}. By Assumption \ref{asp:cmpt}, we can always find a positive scalar $\gamma$ such that $\gamma C_{max,co}$ is contained within $\mathcal{E}$. Since the ellipsoid $\mathcal{E}$ is $\lambda'$-contractive, any states in $\gamma C_{max,co}$ can be steered to $(\lambda')^k \mathcal{E}$ in $k$ steps. Then, for any $\lambda\in (0,1)$, we can find a large enough $N$ such that $(\lambda')^N \mathcal{E} $ is inside $ \lambda \gamma C_{max,co}$, and thus $\gamma C_{max,co}$ is $N$-step $\lambda$-contractive. This intuition is used to compute feasible $\gamma$, $N$ and $\lambda$ in Lemma \ref{lem:contraction} in the next subsection.}
By combining \eqref{eqn:inner_bound_proj} with Lemma \ref{lem:contraction} and Theorem \ref{thm:conv_suff},  the following theorem bounds the convergence of the projection $\Proj{n}(C_{max,p_0+kN})$.
\begin{theorem} \label{thm:conv_asym} 
	Suppose that a system $\Sigma$, a safe set $S_{xu}$ and a preview horizon $p_0$ satisfy Assumptions \ref{asp:asymp_stab}, \ref{asp:cmpt} and \ref{asp:int}. Then, there exists a scalar $ \gamma\in (0,1]$, a positive integer $N$ and a scalar $ \lambda \in [0,1)$ such that the projection $ \Proj{n}(C_{max,p_0+kN})$ satisfies that for $k\leq k_0$,
\begin{align}
	\frac{ \lambda_0}{ \lambda^{k}} C_{max,co}   \subseteq \Proj{n}(C_{max,p_0+kN}) \subseteq    C_{max,co}; \label{eqn:conv_1}
\end{align}
for $k \geq k_0$,
\begin{align}
	\left( 1 - c a^{k} \right) C_{max,co} \subseteq \Proj{n}(C_{max,p_0+kN}) \subseteq C_{max,co},\label{eqn:conv_2}
\end{align}
where $k_0$ is defined by \eqref{eqn:k_0}, $a= (1- \gamma)/(1 - \gamma \lambda)$ and ${c= (1-  \lambda_0/ \lambda^{k_0})a^{k_0}}$.
\end{theorem}

{In the case of $ \lambda_0 = 0$, Theorem \ref{thm:conv_asym} is trivial as $ k_0 = + \infty$ and the leftmost set in \eqref{eqn:conv_1} becomes $\{0\}$. That is why we need Assumption \ref{asp:int2} to enforce $ \lambda_0 > 0$ and exclude this trivial case.} The results in this subsection is summarized by the following corollary.
\begin{corollary} \label{cor:conv_asym} 
Under Assumptions \ref{asp:asymp_stab}, \ref{asp:cmpt}, \ref{asp:int} and \ref{asp:int2},  the projection of $C_{max,p}$ onto the first $n$-coordinates converges to the maximal RCIS $C_{max,co}$ of the disturbance collaborative system $\mathcal{D}(\Sigma)$ in Hausdorff distance, that is 
\begin{align*}
d_{p}=d(\Proj{n}(C_{max,p}), C_{max,co}) \xrightarrow{p \rightarrow \infty} 0. 
\end{align*}
Furthermore, the Hausdorff distance $d_{p}$ satisfies the following inequality: For ${p_0\leq p < p_0 + N(k_0+1)}$,
\begin{align} \label{eqn:conv_rate_asymp_1} 
	d_{p} \leq (1-\lambda_0 \lambda^{-\lfloor (p-p_0)/N \rfloor}) r_{co};
\end{align}
for $p \geq p_0 + N(k_0+1)$,
\begin{align} \label{eqn:conv_rate_asymp_2}
    d_{p} \leq ca^{ \lfloor (p-p_0)/N \rfloor} r_{co},
\end{align}
with  $r_{co}$ is {the radius of the smallest ball centered at $0$ that contains $C_{max,co}$}. The other constants $ k_0$, $N$, $\lambda_0$, $ \lambda$, $a$ and $c$ are the same as in Theorem \ref{thm:conv_asym}.
\end{corollary}

According to \eqref{eqn:conv_rate_asymp_2}, as we increase the preview horizon, the safety regret decays exponentially fast. 
We further define the marginal value $\Delta d_{p}$ of the preview at preview horizon $p$ as the Hausdorff distance between $\Proj{n}(C_{max,p})$ and $\Proj{n}(C_{max,p+N})$. Then, by \eqref{eqn:conv_rate_asymp_2}, for $p\geq p_0+N(k_0+1)$, 
\begin{align}\label{eqn:Dd_conv_rate_asymp_2}
   \Delta d_{p} = d(\Proj{n}(C_{max,p}), \Proj{n}(C_{max,p+N})) \\
   \leq d_{p} + d_{p+N}	\leq c(1+a)a^{ \lfloor (p-p_0)/N \rfloor} r_{co}.
\end{align}
Thus, we show that the marginal value of preview decays exponentially fast as $p$ increases. 
{
\begin{remark}
The upper bounds of $d_p$ and $\Delta d_p$ in \eqref{eqn:conv_rate_asymp_2} and \eqref{eqn:Dd_conv_rate_asymp_2} are functions of the parameters $\gamma$, $N$ and $\lambda$. By the discussion right after  Lemma \ref{lem:contraction}, there are more than one feasible $\gamma$, $N$ and $\lambda$. Conceptually, a tighter upper bound of $d_p$ (or $\Delta d_p$) is obtained by taking the infinium of the upper bounds in \eqref{eqn:conv_rate_asymp_2} (or \eqref{eqn:Dd_conv_rate_asymp_2}) over all the feasible $\gamma$, $N$ and $\lambda$.
\end{remark}
}

\subsection{Estimating the parameters in Theorem \ref{thm:conv_asym}} 
\label{sec:param_asym} 
In this subsection, we show how to numerically obtain an exponentially decaying upper bound of $d_{p}$ predicted by Theorem \ref{thm:conv_asym} for a given system. Specifically, we propose a sequence of optimization programs to estimate the parameters $ \lambda_0$, $ \gamma$, $N$ and $ \lambda$ in Theorem \ref{thm:conv_asym}. To make the computation tractable, we assume that the safe set $S_{xu}$ and the disturbance set $D$ are represented by polytopes. In case where the maximal CIS $C_{max,co}$ and the maximal RCIS $C_{max,p_0}$ are not polytopes, we use a polytopic outer approximation of $C_{max,co}$ and a polytopic inner approximation of $C_{max,p_0}$ instead, which results in a lower estimate of $ \lambda_0$.  Clearly, results in Section \ref{sec:conv_asym} still hold when $ \lambda_0$ is replaced by its lower estimate. There is a rich literature of computing polytopic inner or outer approximations of the maximal RCIS for discrete-time linear systems (see \cite{rakovic2007optimized, rungger2017computing, anevlavis2019computing} for example).   

\textbf{Step 1:}  We check if the origin is a forced equilibrium of $\mathcal{D}(\Sigma)$ that satisfies Assumptions \ref{asp:int} and \ref{asp:int2}. If not, we find a forced equilibrium $(x_{e},u_{e},d_{e})$ satisfying Assumptions \ref{asp:int} and \ref{asp:int2} by the following linear program:
\begin{align} \label{eqn:LP_int} 
	\begin{split}
	 \epsilon^{*} = &\max_{ \epsilon \geq 0, x_{e}, u_{e}, d_{e}}   \epsilon\\
	\text{s.t. }&  Ax_{e} + Bu_{e} + Ed_{e} = x_{e},\\
				&  x_{e}+  \epsilon \mathcal{B}(n)\subseteq \Proj{n}(C_{max,p_0}),\\
				&  (x_e,u_e,d_e) + \epsilon \mathcal{B}(n+m+l) \subseteq S_{xu}\times D.
\end{split}
\end{align}
 Given the $H$-representations of $C_{max,p_0}$, $S_{xu}$ and $D$, the second and third constraints of \eqref{eqn:LP_int} can be easily encoded as linear inequality constraints by iterating vertices of the hypercubes $\mathcal{B}(n)$ and $\mathcal{B}(n+m+l)$.

{By maximizing the cost $\epsilon$ in \eqref{eqn:LP_int}, we push the equilibrium point $(x_e,u_e,d_e)$ more towards the interior of $S_{xu}\times D$ and $\Proj{n}(C_{max,p_0})\times \R^{m+l}$.} When $ \epsilon^{*}>0$, the solution $(x_{e},u_{e},d_{e})$ of \eqref{eqn:LP_int} is a feasible forced equilibrium satisfying Assumptions \ref{asp:int} and \ref{asp:int2}. Then, as highlighted in Remark \ref{rem:origin}, we shift the origin of the state-input-disturbance space to the forced equilibrium computed in \eqref{eqn:LP_int}, that is to shift the sets $S_{xu}$, $D$, $C_{max,co}$ and $C_{max,p_0}$ accordingly. 

\textbf{Step 2:} We want to compute (lower-estimates of) the initial factor $ \lambda_0$.  We first introduce a baseline method, with two steps: 
The first step is to find the maximal $ \lambda$ such that ${ \lambda C_{max,co} \subseteq\mathcal{B}(n)}$. This scalar $ \lambda$ can be computed by a linear program. 
By the construction of \eqref{eqn:LP_int}, we know $ \epsilon^* \mathcal{B}(n) \subseteq \Proj{n}(C_{max,p_0})$, where $ \epsilon^{*}$ is the optimal cost of \eqref{eqn:LP_int}.   Thus, $ \epsilon^{*} \lambda C_{max,co}	\subseteq \epsilon^{*} \mathcal{B}(n) \subseteq \Proj{n}(C_{max,p_0})$. That is, $ \epsilon^{*} \lambda$ provides a lower estimate of $ \lambda_0$. The benefits of this method include (i) whenever \eqref{eqn:LP_int} returns a positive $ \epsilon^{*}$ \footnote{Note that if \eqref{eqn:LP_int} returns $ \epsilon^{*}=0$, there is no need to compute $ \lambda_0$ anymore, as Assumption \ref{asp:int2} cannot be verified. }, the estimated $ \lambda_0$ is guaranteed to be positive, and (ii) the computation is easy. The main drawback is that the estimated $ \lambda_0$ can be very conservative.

Alternatively, according to \eqref{eqn:lambda_0}, the estimation of $ \lambda_0$ can be formulated as a polytope containment problem\cite{sadraddini2019linear}: Let $P_1$ and $P_2$ be two polytopes with $H$-representation $[H_1\ h_1]$ and $[H_2\ h_2]$, where $h_i\in \R^{q_{i}}$ for $i=1$, $2$. {Our goal is to find the maximal $ \lambda$ such that $ \lambda P_1 \subseteq P_2$, which is equivalent to find the minimal $r$ such that $P_1 \subseteq r P_2$.} Then, according to Farkas' Lemma, the minimal $r$ can be obtained by the following linear program:
\begin{align}
	\begin{split} \label{eqn:LP_1} 
	r^{*} = &\min_{ r, \Lambda\in R_{+}^{q_{2}\times q_{1}}} r\\
	\text{s.t. } & \Lambda H_{1} = H_{2}\\
				 &\Lambda h_{1} \leq r h_{2}.
	\end{split}
\end{align}

Thus, by replacing the polytopes $P_1$ and $P_2$ in \eqref{eqn:LP_1} with $C_{max,co}$ and $\Proj{n}(C_{max,p_0})$ (or their polytopic approximations), we obtain an estimate of $\lambda_0$ as the reciprocal of the optimal solution $ r^{*}$ of \eqref{eqn:LP_1}.  

This alternative method returns more accurate $ \lambda_0$ than the baseline. When the $H$-repesentations of $C_{max,co}$ and $\Proj{n}(C_{max,p_0})$ are exact (instead of approximated), $ \lambda_0$ estimated by \eqref{eqn:LP_1} matches the true $ \lambda_0$. But, this method is more time consuming than the baseline. Recall that we only have the $H$-representation of $C_{max,p_0}$, but \eqref{eqn:LP_1} needs the $H$-representation of the projection $\Proj{n}(C_{max,p_0})$ of $C_{max,p_0}$. The projection operation of polytopes in $H$-representation is computationally expensive. 

It is also possible to encode the polytope containment constraint in \eqref{eqn:lambda_0} directly based on the $H$-representations of $ C_{max,co}$ and $C_{max,p_0}$, which enables us to estimate $ \lambda_0$ without the projection step\cite{sadraddini2019linear}: Suppose that the $H$-representations of $C_{max,co}$ and $C_{max,p_0}$ are $[H_1\ h_1]$ and $[H_2\ h_2]$ with $h_i\in \R^{q_i}$, $i=1$, $2$. Then, $ \lambda_0$ can be estimated by the reciprocal of the optimal solution of the following linear program: 
\begin{align}
	\begin{split} \label{eqn:LP_2} 
		\lambda_{0}^{-1} = &\min_{ r, \Gamma\in \R^{n_{p}\times n}, \beta\in \R^{n_{p}}, \Lambda\in R_{+}^{q_{2}\times q_{1}}} r\\
	\text{s.t. } & \Lambda H_{1} = H_{2} \Gamma\\
				 &\Lambda h_{1} \leq r h_{2} + H_2 \beta\\
				 & \mathcal{P} \Gamma = I,\\
				 & \mathcal{P} \beta=0,
	\end{split}
\end{align}
where $n_{p}= n+p_0 l$, $I$ is the $n\times n$ identity matrix, and $\mathcal{P}$ is the projection matrix that maps points in $\R^{n_{p}}$ onto the first $n$ coordinates. The linear program in \eqref{eqn:LP_2} is formulated based on a sufficient condition of polytope containment in \cite{sadraddini2019linear}. Thus, $ \lambda_0$ estimated by  \eqref{eqn:LP_2} is more conservative than $ \lambda_0$ estimated by \eqref{eqn:LP_1}.  

To summarize, we propose three methods to estimate $ \lambda_0$, with different conservativeness and computation cost. As discussed in Section \ref{sec:conv_asym}, the estimate of $ \lambda_0$ is meaningful only when it is positive, which can only be guaranteed by the first two methods. Thus, in practice, one can first obtain a positive baseline estimate of $ \lambda_0$ via the first method, and then select the second or third method, based on the available computation power, for a potentially better estimate of $ \lambda_0$. 

\textbf{Step 3:}  We propose a method to find feasible $ \gamma$, $N$ and $\lambda$, inspired by the proof of Lemma \ref{lem:contraction}. First, compute a $\lambda_{a}$-contractive ellipsoidal CIS of $\mathcal{D}(\Sigma)$ for some $\lambda_{a}$ in $[0,1)$.  Here we formulate a bilinear program\footnote{The constraints in \eqref{eqn:lambda_a} that encode $\sqrt{r}$-contractive ellisoidal CISs can be found in Remark 4.1 of \cite{blanchini1994ultimate} and Section 4.4.2 of \cite{blanchini2008set}} to calculate the minimal possible $ \lambda_{a}$:
\begin{align} \label{eqn:lambda_a} 
		&\lambda_{a}^{2} = \min_{ r\geq 0, Q \succ 0, R_1, R_2} ~ r\\
		&\text{subject to } \nonumber\\
					   &\begin{bmatrix}
	Q & AQ+BR_1+ER_2\\
	(AQ+BR_1+ER_2)^{T} & r Q	
	\end{bmatrix} \succ 0. \nonumber
\end{align}
{If we fix the value of $r$, the nonconvex program in \eqref{eqn:lambda_a} becomes a convex feasibility problem. Also, if the program is feasible at $r=r'$, then the optimal $r \leq r'$. 
Thus, the bilinear program in \eqref{eqn:lambda_a} can be solved by a bisection algorithm, where at each step a convex feasibility problem is solved.} For any stabilizable system,  the optimal value  $ \lambda_{a}^{2}$ of \eqref{eqn:lambda_a} must be smaller than one\footnote{See Proposition 23 of \cite{de2004computation}.}.  Given the optimal solution $ r=\lambda_{a}^{2}$, $Q$, $R_1$ and $R_2$ of \eqref{eqn:lambda_a}, the ellipsoid ${\mathcal{E}(c) = \{ x \mid x^{T}Q^{-1}x \leq c^{2}\}} $ is a $ \lambda_{a}$-contractive CIS of $\mathcal{D}(\Sigma)$ for any $c >0$, and the controller $u=R_1 Q^{-1}x$, $u_{d} = R_2 Q^{-1}x$ is a safe controller for any state $x$ in $\mathcal{E}(c)$.  

Let $c_{0}$ be a scalar such that $(x, R_1Q^{-1}x)$ is in $S_{xu}$ and $R_2 Q^{-1}x $ is in $ D$ for all $x\in \mathcal{E}(c_0)$. Let $c_{out}$ be a scalar such that $C_{max,co} \subseteq \mathcal{E}(c_{out})$. The second step of finding $ \gamma$, $ N$ and $ \lambda$ is to estimate the maximal $ c_0$ and the minimal $c_{out}$. 

Suppose that the Cholesky decomposition of the positive definite matrix $Q$ is ${Q=LL^{T}}$ for some invertible matrix $L\in \R^{n\times n}$.  Then, the ellipsoid $ \mathcal{E}(c)$ can be represented equivalently by $\mathcal{E}(c) =\{c L s \mid s^{T}s \leq 1\}$. By Section 8.4.2 of \cite{boyd2004convex}, the maximal $c_0$ is given by the optimal value of the following linear program:
\begin{align}
	\begin{split} \label{eqn:c_0} 
	&c_0 = \max_{c\geq 0} ~ c\\
	&\text{subject to } \\
					   &c \Vert [L^{T}\ L^{-1}R_1^{T}] H^{T}_{xu,i}\Vert_{2} \leq h_{xu,i}, \quad i\in[q_{xu}]\\ 
					   &c \Vert L^{-1}R_2^{T}H^{T}_{D,j}\Vert_{2} \leq  h_{D,j},\quad j\in [q_{d}],
	\end{split}
\end{align}
where $[H_{xu,i}\ h_{xu,i}]$ is the $i$ th row of the $H$-representation $[H_{xu}\ h_{xu}]$ of $S_{xu}$, and $[H_{D,j}\ h_{D,j}]$ is the $j$ th row of the $H$-repsentation $[H_{D}\ h_{D}]$ of $D$, and $q_{xu}$ and $q_{d}$ are the numbers of rows of $H_{xu}$ and $H_{D}$.  

Suppose that the set of vertices of $C_{max,co}$ is $\mathcal{V}$. Then, the square of the minimal $c_{out}$ is equal to  
   $ c_{out} = \max_{v\in \mathcal{V}} \sqrt{v^{T}Q^{-1}v}$. 
However, it is usually time consuming to compute the vertices of $C_{max,co}$ given its $H$-representation.  An alternative way is to first find the minimal bounding rectangle of $C_{max,co}$ and then use the minimal $c$ such that $\mathcal{E}(c)$ contains the minimal bounding rectangle as a conservative estimate of $c_{out}$. 

Once we have the maximal $c_{0}$ and a feasible $c_{out}$, a set of feasible $ \gamma$, $N$ and $ \lambda$ is given by the following theorem.
\begin{theorem} \label{thm:parameters_asym} 
  Given the $ \lambda_{a}$-contractive ellipsoidal CIS $\mathcal{E}(c)$ of $\mathcal{D}(\Sigma)$, the maximal $c_0$ and a feasible $c_{out}$, let 
  \begin{align}
     \gamma &=  c_0/  c_{out}, \label{eqn:gamma} \\
	 N &= \left\lfloor \frac{\log(\gamma)}{\log ( \lambda_{a})} \right\rfloor+1, \label{eqn:N}\\
	 \lambda &=  \lambda_{a}^{N}/ \gamma. \label{eqn:lambda} 
  \end{align}
Then $ \gamma C_{max,co}$ is an $N$-step $ \lambda$-contractive CIS of $\mathcal{D}(\Sigma)$, with $ \gamma \leq 1$ and $ \lambda < 1$.
\end{theorem}
The intuition behind Theorem \ref{thm:parameters_asym} is: By the definition of $c_0$, the ellipsoid $ \mathcal{E}(c_0)$ is a $\lambda_{a}$-contractive CIS of $\mathcal{D}(\Sigma)$ in $S_{xu}\times D$ and thus is contained in the maximal CIS $C_{max,co}$, as shown in Fig. \ref{fig:c_out}.  Then, for $ \gamma = c_{0}/ c_{out} $, we have $ \gamma C_{max,co}$ contained in the $ \lambda_{a}$-contractive CIS $\mathcal{E}(c_0)$. As $\mathcal{E}(c_0)$ can be controlled to reach an arbitrary small neighborhood of the origin over time (by being $\lambda_{a}$-contractive), for an arbitrary small $\lambda$, $ \gamma C_{max,co}$ is $N$-step $ \lambda$-contractive for a large enough $N$.  

\begin{figure}[]
	\centering
	\includegraphics[width=0.35\textwidth]{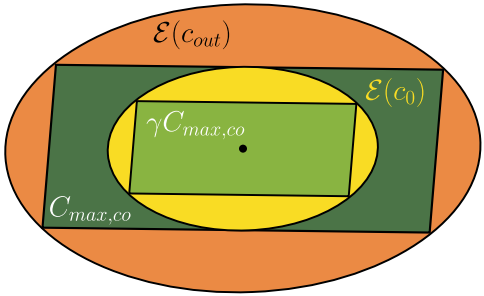}
	\caption{The inclusion relations among $\mathcal{E}(c_{0})$ (yellow), $\mathcal{E}(c_{out})$ (orange), $C_{max,co}$ (dark green) and $ \lambda C_{max,co}$ (light green).}
	\label{fig:c_out}
\end{figure}

\textbf{Step 4 (Refinement):} If the $H$-representation of $C_{max,co}$ is accurate instead of an outer approximation, given the set of feasible parameters $ \gamma$ and $N$ and $ \lambda$ in Theorem \ref{thm:parameters_asym}, we can find another set of feasible parameters which potentially lead to a tighter upper bound of $d_{p}$: Recall that $ \gamma C_{max,co}$ is $N$-step $ \lambda$-contractive.  Let $\gamma^{*}$ be the maximal scalar such that $Pre_{\mathcal{D}(\Sigma)}^{ N} ( \lambda  \gamma C_{max,co}, S_{xu}\times D)$ contains $ \gamma^{*} C_{max,co}$. Here $ \gamma^{*}$ can be solved by a linear program in form of \eqref{eqn:LP_1}. It can be shown that $ \gamma^{*}C_{max,co}$ is $N$-step $ \lambda \gamma/ \gamma^{*}$-contractive.  
Then, we replace the estimated $ \gamma$ and  $ \lambda$ by $ \gamma^{*}$ and $ \lambda \gamma/ \gamma^{*}$.  

The entire procedure (Steps 1-4) for estimating $ \lambda_0$, $\gamma$, $N$ and $ \lambda$ is summarized into Alg. \ref{alg:1}. One can omit the refinement step by removing the last \emph{if} statement in Alg. \ref{alg:1}. In Section \ref{sec:ex2}, we show numerically that the estimated upper bound is much tighter when the refinement is applied. 
\begin{algorithm}
	\caption{Estimating parameters in Theorem \ref{thm:conv_asym}}
\label{alg:1}
\begin{algorithmic}
\State \textbf{input:} $(A,B,E)$, $S_{xu}$, $D$, $C_{max,co}$, $C_{max,p_0}$
\State $(x_{e},u_{e},d_{e}) \gets 0$
\State $ \epsilon^{*}\gets$ solve the LP in \eqref{eqn:LP_int} with $(x_{e},u_{e}, d_{e})$ fixed to $0$
\If{ $ \epsilon^{*} = 0$ \Comment{that is, $(x_{e},u_{e},d_{e})$ is not a forced equilibrium of $\mathcal{D}(\Sigma)$ satisfying Assumptions \ref{asp:int} and  \ref{asp:int2}}
}
\State $(x_{e},u_{e},d_{e}), \epsilon^{*}\gets$ solve the LP in \eqref{eqn:LP_int}
\State $S_{xu}\gets S_{xu}-(x_{e},u_{e})$
\State $D \gets D - d_{e}$ \Comment{shift the origin to $(x_{e},u_{e},d_{e})$}
\EndIf
\State $r^{*}\gets$ solve the LP in \eqref{eqn:LP_1} with $P_1= C_{max,co}$, $ P_2=\mathcal{B}(n)$
\State $ \lambda_0 \gets \epsilon^{*}/r^{*}$ (baseline estimate of $ \lambda_0$)
\If{$H$-representation of $\Proj{n}(C_{max,p_0})$ available} 
\State $r^{*}\gets $ solve the LP in \eqref{eqn:LP_1} with $P_1=C_{max,co}$, $P_2=\Proj{n}(C_{max,p_0})$
\Else 
\State $r^{*}\gets$ solve the LP in \eqref{eqn:LP_2}
\EndIf
\State $ \lambda_0 \gets \max(\lambda_0, 1/r^{*})$ \Comment{Pick the best estimate of $ \lambda_0$} 
\State $\lambda_{a}^{2}, Q, R_1, R_2\gets $ solve the convex program in \eqref{eqn:lambda_a}
\State $c_0\gets$ solve the LP in \eqref{eqn:c_0}
\If{the set $\mathcal{V}$ of vertices of $C_{max,co}$ available}
\State $c_{out} \gets \max_{v\in \mathcal{V}} \sqrt{v^{T}Q^{-1}v}$  
\Else
\State $\mathcal{V}\gets$ vertices of minimal rectangle containing $C_{max,co}$
\State $c_{out} \gets \max_{v\in \mathcal{V}} \sqrt{v^{T}Q^{-1}v}$  
\EndIf
\State $ \gamma, N, \lambda\gets$ RHSs of \eqref{eqn:gamma}, \eqref{eqn:N} and \eqref{eqn:lambda}   
\If {$H$-representation of $C_{max,co}$ is exact} 
\State $C_{N}\gets Pre^{N}_{\mathcal{D}(\Sigma)}( \lambda \gamma C_{max,co}, S_{xu}\times D)$
\State $\gamma^{*}\gets$ solving LP in \eqref{eqn:LP_1} with $P_1=C_{max,co}$, $P_2=C_{N}$
\State $ \gamma \gets \gamma^{*}$, $ \lambda \gets \lambda \gamma/ \gamma^{*}$
\EndIf
\State \Return $ \lambda_0 $, $\gamma $, $ N $, $ \lambda$
\end{algorithmic}
\end{algorithm} 

\subsection{Special case: controllable $\mathcal{D}(\Sigma)$} \label{sec:conv_ctrl} 
According to Corollary \ref{cor:conv_asym}, the convergence of $\Proj{n}(C_{max,p_0+Nk})$ may have two phases, depending on whether $k$ is greater or smaller than $k_0$. In this section, we show a single-phased convergence $\Proj{n}(C_{max,p})$ when the disturbance collaborative system $\mathcal{D}(\Sigma)$ is controllable. Furthermore, we show that Assumption \ref{asp:int2} is no longer required for the convergence of $\Proj{n}(C_{max,p})$.

\begin{assumption} \label{asp:ctrl} 
    The disturbance collaborative system $ \mathcal{D}(\Sigma)$ of the linear system $\Sigma$ is controllable. 
\end{assumption}   

The key observation for controllable disturbance collaborative system is stated by the following lemma.
\begin{lemma} \label{lem:contraction_0} 
   Under Assumptions \ref{asp:cmpt}, \ref{asp:int} and \ref{asp:ctrl}, for any $N \geq n$, there exists a scalar $ \gamma\in (0,1]$ such that $ \gamma C_{max,co}$ is a $N$-step $0$-contractive CIS of the system $\mathcal{D}(\Sigma)$ within the safe set $S_{xu}$, that is    
   \begin{align} \label{eqn:gamma_max} 
	 \gamma C_{max,co} \subseteq
   Pre_{\mathcal{D}(\Sigma)}^{N}(0, S_{xu}). 	
	\end{align}
\end{lemma}
Combining Lemma \ref{lem:contraction_0} with  Theorem \ref{thm:conv_asym}, we bound the convergence of $\Proj{n}(C_{max,p})$ by the following theorem. 
\begin{theorem} \label{thm:conv_ctrl} 
	Suppose that a system $\Sigma$, a safe set $S_{xu}$ and a preview horizon $p_0$ satisfy Assumptions \ref{asp:cmpt}, \ref{asp:int} and \ref{asp:ctrl}. For any $N\geq n$, let $ \gamma_{max}$ be the maximal $\gamma$ such that \eqref{eqn:gamma_max} holds.
	Then, (i) $ \gamma_{max} > 0$ and (ii) the projection $ \Proj{n}(C_{max,p_0+kN})$ satisfies that for $k\geq 0$,
\begin{align}
& C_{max,co} \supseteq \Proj{n}(C_{max,p_0+kN})\nonumber \\
	& \supseteq \left( 1 - (1-  \lambda_0) \left( 1- \gamma_{max}\right)^{k} \right) C_{max,co}. \label{eqn:conv_3}
\end{align}
\end{theorem}
\begin{proof}(Sketch)
For the $N$-step $0$-contractive CIS $\gamma C_{max,co}$ in Lemma \ref{lem:contraction_0}, $ k_0$ in  \eqref{eqn:k_0} is $0$ due to $ \lambda =0$. Then Theorem \ref{thm:conv_asym} with $k_0=0$ implies Theorem \ref{thm:conv_ctrl}. 
\end{proof}

It is obvious that the right hand of \eqref{eqn:conv_3} converges to $C_{max,co}$ even if $ \lambda_0=0$. Thus, here we do not need Assumption \ref{asp:int2} to make $ \lambda_0>0$.

\begin{corollary} \label{cor:conv_ctrl} 
Under Assumptions \ref{asp:cmpt}, \ref{asp:int} and \ref{asp:ctrl},  the projection of $C_{max,p}$ onto the first $n$-coordinates converges to the maximal RCIS $C_{max,co}$ of the disturbance collaborative system $\mathcal{D}(\Sigma)$ in Hausdorff distance, that is 
\begin{align*}
d_{p}=d(\Proj{n}(C_{max,p}), C_{max,co}) \xrightarrow{p \rightarrow \infty} 0. 
\end{align*}
Furthermore, the Hausdorff distance $d_{p}$ satisfies the following inequality: For $p\geq p_0$,
\begin{align} \label{eqn:conv_rate_asymp_3}
    d_{p} \leq  \left( 1- \lambda_0\right)\left(1- \gamma_{max} \right)^{ \lfloor (p-p_0)/N\rfloor} r_{co},
\end{align}
with $N$ and $ \gamma_{max}$ in Theorem \ref{thm:conv_ctrl}, and $r_{co}$ {the radius of the smallest ball centered at $0$ that contains $C_{max,co}$}.  
\end{corollary}
Compared with the upper bound on $d_{p}$ in Corollary \ref{cor:conv_asym}, the upper bound in \eqref{eqn:conv_rate_asymp_3} does not have the burn-in time $k_0$ and is easier to compute, which is shown next. 

We propose an algorithm that numerically calculates the parameters $ \lambda_0$ and $ \gamma_{max}$ in Theorem \ref{thm:conv_ctrl} for controllable $\mathcal{D}(\Sigma)$. The same as in Section \ref{sec:param_asym}, we assume that the safe set $S_{xu}$ and the disturbance set $D$ are polytopes, and the sets $C_{max,co}$ and $C_{max,p_0}$ are replaced by their polytopic outer/inner  approximations when they are not polytopes. 

\textbf{Step 1:} We first check if the origin as a forced equilibrium satisfies Assumption \ref{asp:int}. {If not, we find a feasible forced equilibrium via a linear program modified from \eqref{eqn:LP_int}, where we replace the constraint $x_{e} + \epsilon \mathcal{B}(n) \subseteq \Proj{n}(C_{max,p_0})$ in \eqref{eqn:LP_int} by $x_{e}\subseteq \Proj{n}(C_{max,p_0})$. 

When the optimal value $ \epsilon^{*}>0$, the optimal solution $(x_{e},u_{e},d_{e})$ of the modified linear program is a feasible forced equilibrium satisfying Assumptions \ref{asp:int}. Then, by Remark \ref{rem:origin}, we shift the origin of the state-input-disturbance space to this forced equilibrium.}

\textbf{Step 2:} We compute $ \lambda_0$ by the second or the third methods of estimating $ \lambda_{0}$ proposed in Section \ref{sec:param_asym}. Since $ \lambda_0>0$ is no longer necessary, one may prefer the third method for computational efficiency.

\textbf{Step 3:} We select a step size $N$ and estimate the corresponding $ \gamma_{max}$. Since $S_{xu}$ and $D$ are polytopes, the $N$-step backward reachable set $Pre_{\mathcal{D}(\Sigma)}^{N}(\{0\},S_{xu}\times D )$ is a polytope, whose $H$-representation can be easily computed. Then, solving the maximal $ \gamma$ satisfying \eqref{eqn:gamma_max}, that is $\gamma_{max}$,  is again a polytope containment problem. Thus, $ \gamma_{max}$ is equal to the reciprocal of the optimal value $r^{*}$ of \eqref{eqn:LP_1} with $P_1= C_{max,co}$ and $P_2= Pre_{\mathcal{D}(\Sigma)}^{N}(0,S_{xu}\times D)$. 

The procedure (Steps 1-3) for estimating $ \lambda_0$ and $\gamma_{max}$ is summarized into Alg. \ref{alg:2}. Note that different from Alg. \ref{alg:1}, {one can freely select} the step size $N$.  By Lemma \ref{lem:contraction_0}, for any $N\geq n$, Alg. \ref{alg:2} is guaranteed to find a nonzero $\gamma_{max}$. In Section \ref{sec:ex2}, we show numerically how different $N$ affects the estimated upper bound of $d_{p}$. 
\begin{algorithm}
\caption{Estimating parameters in Theorem \ref{thm:conv_ctrl} }
\label{alg:2}
\begin{algorithmic}
\State \textbf{input:} $N$, $(A,B,E)$, $S_{xu}$, $D$, $C_{max,co}$, $C_{max,p_0}$
\State $(x_{e},u_{e},d_{e}) \gets 0$
\State $ \epsilon^{*}\gets$ solve the LP modified from \eqref{eqn:LP_int} with $(x_{e},u_{e}, d_{e})$ fixed to $0$
\If{ $ \epsilon^{*} = 0$ \Comment{that is, $(x_{e},u_{e},d_{e})$ is not a forced equilibrium of $\mathcal{D}(\Sigma)$ satisfying Assumption \ref{asp:int}}
}
\State $(x_{e},u_{e},d_{e}), \epsilon^{*}\gets$ solve the LP modified from \eqref{eqn:LP_int}
\State $S_{xu}\gets S_{xu}-(x_{e},u_{e})$
\State $D \gets D - d_{e}$ \Comment{shift the origin to $(x_{e},u_{e},d_{e})$}
\EndIf
\If{$H$-representation of $\Proj{n}(C_{max,p_0})$ available} 
\State $r^{*}\gets $ solve the LP in \eqref{eqn:LP_1} with $P_1=C_{max,co}$, $P_2=\Proj{n}(C_{max,p_0})$
\Else 
\State $r^{*}\gets$ solve the LP in \eqref{eqn:LP_2}
\EndIf
\State $ \lambda_0 \gets 1/r^{*}$
\State $C_{N} \gets Pre^{N}_{\mathcal{D}(\Sigma)}(\{0\}, S_{xu}\times D)$
\State $r^{*}\gets$ solve LP in \eqref{eqn:LP_1} with $P_1 = C_{max,co}$, $P_2=C_{N}$. 
\State $\gamma_{max} \gets 1/r^{*}$
\State \Return $ \lambda_0 $, $\gamma_{max} $
\end{algorithmic}
\end{algorithm} 

\subsection{On the finite-time convergence of $\pi_{[1,n]}(C_{max,p})$} \label{sec:finite-conv} 
In practice, $\Proj{n}(C_{max,p})$ may converge to $C_{max,co}$ in finite preview horizon $p$, which is not reflected by the exponentially decaying upper bounds in Theorems \ref{thm:conv_asym} and \ref{thm:conv_ctrl}. In this subsection, we propose a simple algorithm to detect the potential finite-time convergence of $\Proj{n}(C_{max,p})$. 

Suppose that the projection $\Proj{n}(C_{max,p_0})$ is known for some $p_0$. According to Lemma \ref{lem:pre_expand}, $\Proj{n}(C_{max,p})$ is equal to $C_{max,co}$ for $p=p_0+k$ if
\begin{align}
   Pre^{k}_{\mathcal{D}(\Sigma)}(\Proj{n}(C_{max,p_0}),S_{xu}\times D)=C_{max,co}. 
\end{align}
Based on the above sufficient condition, we propose Alg. \ref{alg:3} to detect the finite-time convergence of $\Proj{n}(C_{max,p_0})$. 
\begin{algorithm}
\caption{Detecting finite-time convergence}
\label{alg:3}
\begin{algorithmic}
\State {\bf input:} $C_{max,co}$, $\Proj{n}(C_{max,p_0})$, $k_{max}$ 
\State $k \gets 0$, $C_0 \gets \Proj{n}(C_{max,p_0})$, $\overline{p}\gets  \infty$
\If {$C_{0} = C_{max,co}$} $\overline{p} \gets p_0$ 
\EndIf
\While{$k < k_{max}$}
\State  $k \gets	k+1$, $C_{k} \gets  Pre_{\mathcal{D}(\Sigma)}(C_{k-1}, S_{xu}\times D)$, 
\If {$C_{max,co} = C_{k}$} $\overline{p} \gets p_0+k$; break
\EndIf
\EndWhile
\State \Return $\overline{p}$, $\{C_{k}\}_{k=0}^{\min( \overline{p}, k_{max})} $ 
\end{algorithmic}
\end{algorithm} 
Note that we set a maximal iteration number $k_{max}$ to guarantee the termination of Alg. \ref{alg:3}. In case where Alg. \ref{alg:3} returns a finite number $\overline{p}$, we know thta $\Proj{n}(C_{max,\overline{p}})$ must be equal to $C_{max,co}$.  
In terms of the computational cost, given the $H$-representation of $\Proj{n}(C_{max,p_0})$, the $k$-step backward reachable set of  $\Proj{n}(C_{max,p_0})$ with respect to $\mathcal{D}(\Sigma)$ is much cheaper to compute than $\Proj{n}(C_{max,p+k})$, since the dimensions of the systems $\mathcal{D}(\Sigma)$ are independent of the preview horizon $p$.  Thus, Alg. \ref{alg:3} is more tractable than directly checking if $\Proj{n}(C_{max,p})=C_{max,co}$ for $p\geq 0$. 

As a side product, the Hausdorff distance between the backward reachable set $C_{k} $ computed in Alg. \ref{alg:3} and $C_{max,co}$ gives an upper bound of $d_{p_0+k}$ for $k=0, \cdots, \min(\overline{p}-p_0, k_{max})$,  tighter than those obtained by Algs. \ref{alg:1} and \ref{alg:2}. 
However, Alg. \ref{alg:3} is more computationally expensive than Algs. \ref{alg:1} and \ref{alg:2} due to the iterative backward reachable set computation (and the Hausdorff distance computation). Another drawback is that Alg. \ref{alg:3} can only provide the upper bounds of $d_{p}$ for finitely many $p$. 

\section{Model Predictive Control with Preview} \label{sec:mpc} 
In this section, we study the impact of preview on the feasible domain of constrained MPC under different recursive feasibility constraints. We consider the class of discrete-time linear systems $\Sigma$ as in  \eqref{eqn:sys_linear} with $p$-step preview. Let the MPC {prediction horizon} be the preview horizon $p$ and the constraints on state and input be the safe set $S_{xu}$ of $\Sigma$ for simplicity. Given the current state $x_0$ and preview information $d_{0:p-1}$, a standard optimization problem to be solved by MPC at each time step is:
\begin{align}
	\begin{split} \label{eqn:mpc} 
	\min_{x_{1:p}, u_{0:p-1}} & \ell_{F}(x_{p},u_{p-1}) + \sum^{p-1}_{t=1} \ell_{t}(x_{t}, u_{t-1}) \\
	\text{s.t.} &~  x_{t} = A x_{t-1} + Bu_{t-1} + Ed_{t-1},\\
				& (x_{t-1},u_{t-1})\in S_{xu}, \forall t\in [p],\\
				& \text{RFC},
	\end{split}
\end{align}
where $\ell_{t}$ and $\ell_{F}$ are the stage cost at time $t$ and the final cost, and  RFC is a placeholder for constraints that guarantee the recursive feasibility of \eqref{eqn:mpc}. 

\begin{definition}
 The feasible domain $\mathcal{F}$ of the MPC in \eqref{eqn:mpc} is the set of initial conditions $(x_0,d_{0:p-1})$ that satisfy the constraints in \eqref{eqn:mpc}.   
\end{definition}
 To guarantee the recursive feasibility of the MPC, we need to impose certain recursive feasibility constraints RFC in \eqref{eqn:mpc} such that the feasible domain $\mathcal{F}$ is an RCIS of the $p$-augmented system $\Sigma_{p}$ in $S_{xu,p}$. Depending on the choice of the RFC, the size of feasible domain can be very different. Here we study two RFCs.

First, let RFC in  \eqref{eqn:mpc} be 
\begin{align} \label{eqn:RFC1} 
	(x_1, d_{2:p-1},d_{p})\in C_{max,p}, \forall d_{p}\in D.
\end{align}
The constraints in \eqref{eqn:RFC1} is the least conservative RFC, since the corresponding feasible domain is the maximal RCIS $C_{max,p}$ of $\Sigma_{p}$ in $S_{xu,p}$, the largest feasible domain any RFC can produce.  

Since the maximal RCIS $C_{max,p}$ is hard to compute as $p$ increases, a more common RFC in practice is to force the final state $x_{p}$ to be in an RCIS $C$ of the original system $\Sigma$ in $S_{xu}$, that is
\begin{align} \label{eqn:RFC2} 
    x_{p}\in C. 
\end{align}
The following theorem characterizes the relation between the terminal constraint set $C$ and the corresponding feasible domain.
\begin{theorem} \label{thm:fd}  
The feasible domain, denoted by $\mathcal{F}_{p}(C)$, of the MPC corresponding to the RFC as in \eqref{eqn:RFC2} is the $p$-step backward reachable set of $C\times D^{p}$ for $\Sigma_{p}$ in $S_{xu,p}$. That is,   
\begin{align} \label{eqn:fd} 
  \mathcal{F}_{p}(C) = Pre_{\Sigma_{p}}^{p}(C\times D^{p}, S_{xu,p}).
\end{align}
\end{theorem}
\begin{proof}
Combining constraints in \eqref{eqn:mpc} and \eqref{eqn:RFC2}, it is easy to check that $\mathcal{F}_{p}(C) \subseteq Pre_{\Sigma_{p}}^{p}(C\times D^{p}, S_{xu,p})$. It is remained to show the other inclusion direction.

Suppose $(x_0,d_{0:p-1})\in Pre_{\Sigma_{p}}^{p}(C\times D^{p}, S_{xu,p})$. Then, for arbitrary $d_{p:2p-1}\in D^{p}$, there exists $u_{0:p-1}$ such that for all $t=0, \cdots, p-1$, 
\begin{align}
&(x_{t},d_{t:t+p-1},u_{t})\in S_{xu,p}, \\
&(x_{p}, d_{p:2p-1})\in C\times D^{p}.
\end{align}
That is, $(x_{t},u_{t})\in S_{xu}$ for $t=0, \cdots, p-1$ and $x_{p}\in C$, which implies $(x_{0},d_{0:p-1})\in \mathcal{F}_{p}(C)$.
\end{proof}
By \eqref{eqn:fd}, $\mathcal{F}_{p}(C)$ is an RCIS of $\Sigma_{p}$ in $S_{xu,p}$ and thus is a subset of $C_{max,p}$. Intuitively, $\mathcal{F}_{p}(C)$ is the set of states where the system is guaranteed to stay in the safe set indefinitely even if no preview is available after the first $p$ steps, and thus is more conservative than $C_{max,p}$.

We want to compare how the gap between $\mathcal{F}_{p}(C)$ and $C_{max,p}$ changes as $p$ increases.  
Since the dimensions of the two feasible domains $C_{max,p}$ and $\mathcal{F}_{p}(C)$ increase with $p$, a direct comparison is intractable. Similar to Section \ref{sec:convergence}, we instead compare the gap between the projections of $\mathcal{F}_{p}(C)$ and $C_{max,p}$ onto the first $n$ dimensions. 
\begin{lemma}\label{lem:Fp_pre}
   The projection of $\mathcal{F}_{p}(C)$ is equal to the $p$-step backward reachable set of $C$ with respect to $\mathcal{D}(\Sigma)$ in $S_{xu}\times D$. That is,
$\Proj{n}(\mathcal{F}_{p}(C))= Pre_{\mathcal{D}(\Sigma)}^{p}(C, S_{xu}\times D). $
\end{lemma}
\begin{proof}
Let $x(0)\in \Proj{n}(\mathcal{F}_{p}(C))$.  There exists $d_{1:p}(0)\in D^{p}$ such that $(x(0),d_{1:p}(0))\in \mathcal{F}_{p}(C)$. Since by Theorem \ref{thm:fd} $\mathcal{F}_{p}(C) = Pre_{\Sigma_{p}}^{p}(C\times D^{p},S_{xu,p})$, for the system $\Sigma_{p}$ as in \eqref{eqn:sys_p}, there exist inputs $\{u(t)\}_{t=0}^{p-1} $ such that $(x(t),u(t))\in S_{xu}$ for $t=0, \cdots, p-1$ and $x(p)\in C$ for all possible $\{d(t)\}_{t=0}^{p} \subseteq D$. That implies $x(0)\in Pre_{\mathcal{D}(\Sigma)}^{p}(C, S_{xu}\times D)$. Thus, $\Proj{n}(\mathcal{F}_{p}(C))$ is a subset of $ Pre_{\mathcal{D}(\Sigma)}^{p}(C, S_{xu}\times D)$.

Next, we want to prove the other direction. Let $x(0)\in  Pre_{\mathcal{D}(\Sigma)}^{p}(C, S_{xu}\times D)$. Then, for the system $\mathcal{D}(\Sigma)$ as in \eqref{eqn:sys_dis} , there exist inputs $\{u(t)\}_{t=0}^{p-1}$ and $\{u_{d}(t)\}_{t=0}^{p-1} \subseteq D$ such that $(x(t),u(t))\in S_{xu}$ for $t=0, \cdots, p-1$ and $x(p)\in C$. That implies $(x(0),u_{d}(0), \cdots, u_{d}(p-1))$ is contained by $Pre_{\Sigma_{p}}^{p}(C\times D, S_{xu,p})= \mathcal{F}_{p}(C)$ by Theorem \ref{thm:fd}.  Thus, $x(0)\in \Proj{n}(\mathcal{F}_{p}(C))$.
\end{proof}

By Lemma \ref{lem:Fp_pre}, it is clear that the projection $\Proj{n}(\mathcal{F}_{p}(C))$ is a CIS of $\mathcal{D}(\Sigma)$ in $S_{xu}\times D$. Recall that $C_{max,co}$ is the maximal CIS of $\mathcal{D}(\Sigma)$ in $S_{xu}\times D$. Thus, we have
\begin{align} \label{eqn:Fp_sdw} 
 \Proj{n}(\mathcal{F}_{p}(C)) \subseteq \Proj{n}(C_{max,p}) \subseteq C_{max,co}.
\end{align}
Following similar steps in Section \ref{sec:convergence}, we can show the convergence rate of the projection $\Proj{n}(\mathcal{F}_{p}(C))$ over the prediction horizon $p$.  First, we  modify Assumptions \ref{asp:int} and \ref{asp:int2} by replacing $\Proj{n}(C_{max,p_0})$ with the set $C$, and call the modified assumptions Assumption \ref{asp:int}' and \ref{asp:int2}'. Let $ \lambda_0$ be the maximal $ \lambda$ such that $\lambda C_{max,co} \subseteq C$, which is greater than $0$ under Assumption \ref{asp:int2}'. Then, the following theorem shows that $\Proj{n}(\mathcal{F}_{p}(C))$ converges to $C_{max,co}$ exponentially fast.
\begin{theorem} \label{thm:Fp}  
Under Assumptions \ref{asp:asymp_stab}, \ref{asp:cmpt}, \ref{asp:int}' and \ref{asp:int2}' (or Assumptions \ref{asp:cmpt}, \ref{asp:int}' and \ref{asp:ctrl}), the projection $\Proj{n}(\mathcal{F}_{p}(C))$ converges to the maximal RCIS $C_{max,co}$ of $\mathcal{D}(\Sigma)$ in Hausdorff distance, that is 
\begin{align}
d(\Proj{n}(\mathcal{F}_{p}(C)), C_{max,co}) \xrightarrow{p \rightarrow \infty}0.	
\end{align}
Furthermore, there exist some constants $c>0$ and $a\in [0,1)$ such that for $p\geq 0$, 
\begin{align} \label{eqn:Fp_conv} 
 d(\Proj{n}(\mathcal{F}_{p}(C)), C_{max,co})  \leq  c a^{p}.  
\end{align}
\end{theorem}
{The proof of Theorem  \ref{thm:Fp} is similar to the proofs of Theorems \ref{thm:conv_asym} and \ref{thm:conv_ctrl} (by replacing $C_{max,p}$ with $\mathcal{F}_{p}(C)$) and thus is omitted.}
\begin{remark}
{Due to the similarity between Theorem \ref{thm:Fp} and Theorems \ref{thm:conv_asym}, \ref{thm:conv_ctrl}}, one may expect that the upper bound on $d(\Proj{n}(\mathcal{F}_{p}(C)), C_{max,co})$ is in form of \eqref{eqn:conv_rate_asymp_1}, \eqref{eqn:conv_rate_asymp_2} or \eqref{eqn:conv_rate_asymp_3}. It is indeed the case, but we relax those tighter bounds to the RHS of \eqref{eqn:Fp_conv} for simplicity. The tighter upper bounds can be retrieved from Section \ref{sec:convergence} by replacing $\Proj{n}(C_{max,p})$ there with $C$.
\end{remark}

According to \eqref{eqn:Fp_sdw} and Theorem \ref{thm:Fp}, we have
	\begin{align*}
	    & d(\Proj{n}(\mathcal{F}_{p}(C)), \Proj{n}(C_{max,p}))\\
	 & \leq d(\Proj{n}(\mathcal{F}_{p}(C)), C_{max,co}) \leq  ca^{p}.
	\end{align*}
That is, the gap between the projections of $\mathcal{F}_{p}(C)$ and $C_{max,p}$ decays exponentially fast under the conditions in Theorem \ref{thm:Fp}. Thus, even if $\mathcal{F}_{p}(C)$ is more conservative than $C_{max,p}$, as the prediction horizon $p$ is increased, this conservativeness decays fast. Thus, in practice, the RFC in \eqref{eqn:RFC2} with a long enough prediction horizon $p$ is usually a good choice. Also, the constants  $c$ and $a$ in Theorem \ref{thm:Fp} can be estimated by some modified Algs. \ref{alg:1} and \ref{alg:2} that replace $\Proj{n}(C_{max,p_0})$ in both algorithms by $C$. Once $c$ and $a$ are obtained, one can simply compute $ca^{p}$ to quantitatively evaluate the conservativeness of  the RFC in \eqref{eqn:RFC2} with respect to the maximal feasible domain $C_{max,p}$ for any given prediction horizon $p$. Besides, similar to Section \ref{sec:finite-conv}, the potential finite-time convergence of the projections of  $\mathcal{F}_{p}(C)$ can be detected by a modified Alg. \ref{alg:3} that replaces $\Proj{n}(C_{max,p_0})$ with $C$. 

\section{Examples} \label{sec:ex} 
\subsection{One-dimensional systems}
Consider the $1$-dimensional system\cite{liu2021value} 
\begin{align} \label{eqn:sys_1d} 
    \Sigma: x(t+1) = a x(t) + u(t) + d(t),
\end{align}
with $x(t)$, $u(t)\in \R$ and $d(t)\in [-\overline{d},\overline{d}]$. The safe set $S_{xu}= [-\overline{x},\overline{x}]\times [-\overline{u}, \overline{u}]$. 

Suppose that the parameters $a$, $\overline{x}$, $\overline{d}$, $\overline{u}$ and $p$ satisfy $a>1$, $ \overline{x} \geq ( \overline{u}+ \overline{d})/ (a-1)$ and $a^{p-1} \overline{u} \geq \overline{d}$. Then, it is shown in \cite{liu2021value} that the maximal RCIS $C_{max,p}$ of the $p$-augmented system of $\Sigma$ within the augmented safe set $[-\overline{x},\overline{x}]\times [-\overline{d},\overline{d}]^{p}\times [-\overline{u},\overline{u}]$ is
\begin{align}
	C_{max,p} = \bigg\{(x,d_{1:p}) ~\bigg\vert ~ \vert d_{i} \vert\leq  \overline{d},\forall i\in \{1, \cdots, p\}, \nonumber \\
	\vert x+ \sum_{i=1}^{p} \frac{d_{i}}{a^{i}}\vert \leq \frac{ \overline{u}-\overline{d}/a^{p}}{a-1} \bigg\}.
\end{align}
The projection of the maximal controlled invariant set onto the first coordinate is 
\begin{align} \label{eqn:proj_1d} 
\Proj{1}(C_{max,p}) = \left[- \frac{\overline{u}+\overline{d}-2 \overline{d}/a^{p}}{a-1}, \frac{\overline{u}+\overline{d}-2 \overline{d}/a^{p}}{a-1}\right].
\end{align}
The disturbance-collaborative system with respect to \eqref{eqn:sys_1d} is
\begin{align}
    \mathcal{D}(\Sigma): x(t+1) = ax(t) + u,
\end{align}
with the safe set $S_{xu,co}=[-\overline{x},\overline{x}]\times [-\overline{u}-\overline{d}, \overline{u}+\overline{d}]$. The maximal controlled invariant set $C_{max,co}$ of $\mathcal{D}(\Sigma)$ in the safe set is ${[-(\overline{u}+\overline{d})/(a-1), (\overline{u}+\overline{d})/(a-1)]}$.  Thus, for any $p$ such that $a^{p-1} \overline{u} \geq \overline{d}$, the Hausdorff distance between $\Proj{1}(C_{max,p})$ and $C_{max,co}$ satisfies 
\begin{align} \label{eqn:upper_bnd_true} 
    d(\Proj{1}(C_{max,p}), C_{max,co}) = \frac{2\overline{d}}{(a-1) a^{p}}, 
\end{align}
which decays to $0$ with exponential rate $1/a$ as $p$ goes to infinity. 

Note that this one-dimensional system is simple enough for us to solve Alg. \ref{alg:2} by hand: Since $Pre_{\mathcal{D}(\Sigma)}(0,S_{xu,co}) = [-(\overline{u}+\overline{d})/a, (\overline{u}+\overline{d})/a]$, the value of $ \gamma_{max}$ in Theorem \ref{thm:conv_ctrl} is $(1-1/a)$. We pick any $p_0$ satisfying $a^{p_0-1}\overline{u} \geq  \overline{d}$. Then, by \eqref{eqn:proj_1d}, $ \lambda_0$ is equal to
\begin{align}
    \lambda_0 = 1 - \frac{2 \overline{d}/a^{p_0}}{(\overline{u}+\overline{d})}.
\end{align}
The radius {$r_{co}$ of the smallest ball centered at $0$ containing $C_{max,co}$} is $(\overline{u}+\overline{d})/(a-1)$. Plugging $ \gamma_{max}$, $ \lambda_0$ into \eqref{eqn:conv_rate_asymp_3}, for all $p\geq p_0$, the Hausdorff distance between $\Proj{1}(C_{max,p})$ and $C_{max,co}$ satisfies 
\begin{align} \label{eqn:upper_bnd_est} 
   d(\Proj{1}(C_{max,p}), C_{max,co})\leq  \frac{2\overline{d}}{(a-1)a^{p}}.
\end{align}
Comparing the right hand sides of \eqref{eqn:upper_bnd_true} and \eqref{eqn:upper_bnd_est}, the upper bound of the Hausdorff distance $d_{p}$ obtained by Alg. \ref{alg:2}  is equal to the actual value in this toy example.

\subsection{Two-dimensional system with random safe set} \label{sec:ex2} 
We consider the following $2$-dimensional system $\Sigma$:
\begin{align}  \label{eqn:ex2} 
   \Sigma: x^{+} = \begin{bmatrix}
   1.5 & 1\\ 0 & 1.1	
   \end{bmatrix} x + 
   \begin{bmatrix}
   0 \\ 1	
   \end{bmatrix}u +
   \begin{bmatrix}
   1 \\ 1	
   \end{bmatrix}d,
\end{align}
with $x\in \R^{2}$, $u\in \R$ and $d\in [-0.3,0.3]$.  The safe set $S_{xu}$ is randomly generated, shown by the dark blue polytope in Fig. \ref{fig:ex2_visual}. Assume that the preview on $d$ is available.  We select $p_0=1$. The projection of  $C_{max,p_0}$ onto $\R^{2}$ is shown by the yellow polytope in Fig. \ref{fig:ex2_visual}. The projections of $C_{max,p_0+k}$ for $k=1, \cdots, 8$ are shown by the nested cyan polytopes in Fig. \ref{fig:ex2_visual}. The set $C_{max,co}$ is equal to the projection of $C_{max,p}$ with $p=9$, shown by the largest cyan polytope in Fig. \ref{fig:ex2_visual}. 
\begin{figure}[]
	\centering
	\includegraphics[width=0.4\textwidth]{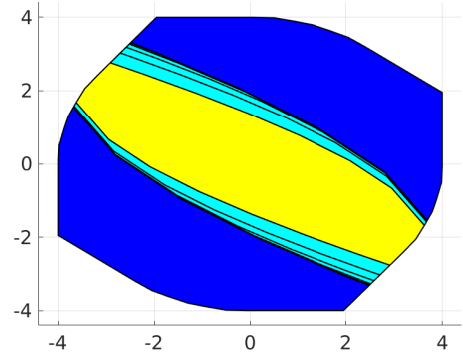}
	\caption{The projections of $C_{max,p_0}$ (yellow polytope), $C_{max,p_0+k}$ (cyan polytopes) with $k=1, \cdots, 8$ and the safe set (dark blue polytope).}
	\label{fig:ex2_visual}
\end{figure}

Since $\Sigma$ is controllable, we apply both Algs. \ref{alg:1} and \ref{alg:2} to this example. 
The parameters estimated by Alg. 1 (with or without Step 4) and Alg. 2 ($N=1$ or $8$) are listed in Table \ref{tab:ex2}. Plugging the parameters in Table \ref{tab:ex2} into \eqref{eqn:conv_rate_asymp_1}, \eqref{eqn:conv_rate_asymp_2} and \eqref{eqn:conv_rate_asymp_3}, we obtain four upper bounds of the Hausdorff distance $d_{p}$ in \eqref{eqn:dp}, as depicted in Fig. \ref{fig:ex2_curve}. 
Note that the red curve in Fig. \ref{fig:ex2_curve} lies below the blue curve, showing that the refinement step in Alg. \ref{alg:1} indeed helps us find a tighter upper bound on $d_{p}$. 
 The purple curve obtained by Alg. \ref{alg:2} is coarser than the other curves due to a larger step size $(N=8)$, but is also decaying faster than the others as $p$ increases. This observation implies that selecting a larger $N$ in Alg. \ref{alg:2} may lead to a coarser but faster decaying upper bound. 
Among all the upper bounds in Fig. \ref{fig:ex2_curve},  Alg. \ref{alg:1} with Step 4 and Alg. \ref{alg:2} ($N=1$) find the tightest upper bounds (red and yellow curves) when $p$ is small; Alg. \ref{alg:2} ($N =8$) finds the tightest bound (purple curve in Fig. \ref{fig:ex2_curve}) when $p$ is large.  

We also run Alg. \ref{alg:3} with $k_{max}=50$, which terminates with $\overline{p}= \infty$. That indicates $\Proj{n}(C_{max,p})$ may not converge to $C_{max,co}$ for a finite $p$. The upper bound of $d_{p}$ obtained as the side product of Alg. \ref{alg:3} is the tightest, coinciding with the actual $d_{p}$ shown by the green curve in Fig. \ref{fig:ex2_curve}. 
The computation time of the three algorithms is listed in the first column of Table \ref{tab:time}. 
\begin{table}[]
	\centering
	\caption{Parameters estimated by Algs. \ref{alg:1} and \ref{alg:2}}
	\label{tab:ex2}
	\begin{tabular}{c|cccc}
		\hline
		& $ \lambda_0$ & $ \gamma$ or $\gamma_{max}$ & $N$ & $ \lambda$\\
		\hline
		Alg. \ref{alg:1} w.o. Step 4 & $0.6550$ & $0.0402$ & $1$ & $0.0011$\\	
		Alg. \ref{alg:1} w. Step 4 &$0.6550$ & $0.0752$ & $1$ & $5.737\times 10^{-4}$\\	
		Alg. \ref{alg:2} ($N=1$) & $0.6550$& $0.0752$ & $1$ & $0$\\	
		Alg. \ref{alg:2} ($N=8$) &$0.6550$ & $0.9794$ & $8$ & $0$\\	
		\hline
	\end{tabular}
\end{table}

\begin{table}[]
	\centering
	\caption{Computation time of Algs. \ref{alg:1}, \ref{alg:2} and \ref{alg:3}}
	\label{tab:time}
	\begin{tabular}{c|ccc}
		\hline
		 Time ($s$) & $2$D system & Lane-keeping & Biped\\ 
		\hline
		Alg. \ref{alg:1} w.o. Step 4 & $3.58$ & $3.86$ & $4.42$\\ 
		Alg. \ref{alg:1} w. Step 4 & $4.13$ & $727.09$ & $8.68$\\ 	
		Alg. \ref{alg:2} ($N=n$) & $0.58$ & $13.22$ & $0.84$\\ 	
		Alg. \ref{alg:2} ($N=8$) & $3.66$ & $227.37$ & $2.41$\\ 	
		Alg. \ref{alg:3} ($k_{max}=50$) & $14.29$ & $573.33$ & $188.55$\\
		\hline
	\end{tabular}
\end{table}
\begin{figure*}[]
	\centering
     \begin{subfigure}[b]{0.32\textwidth}
		\includegraphics[width=\textwidth]{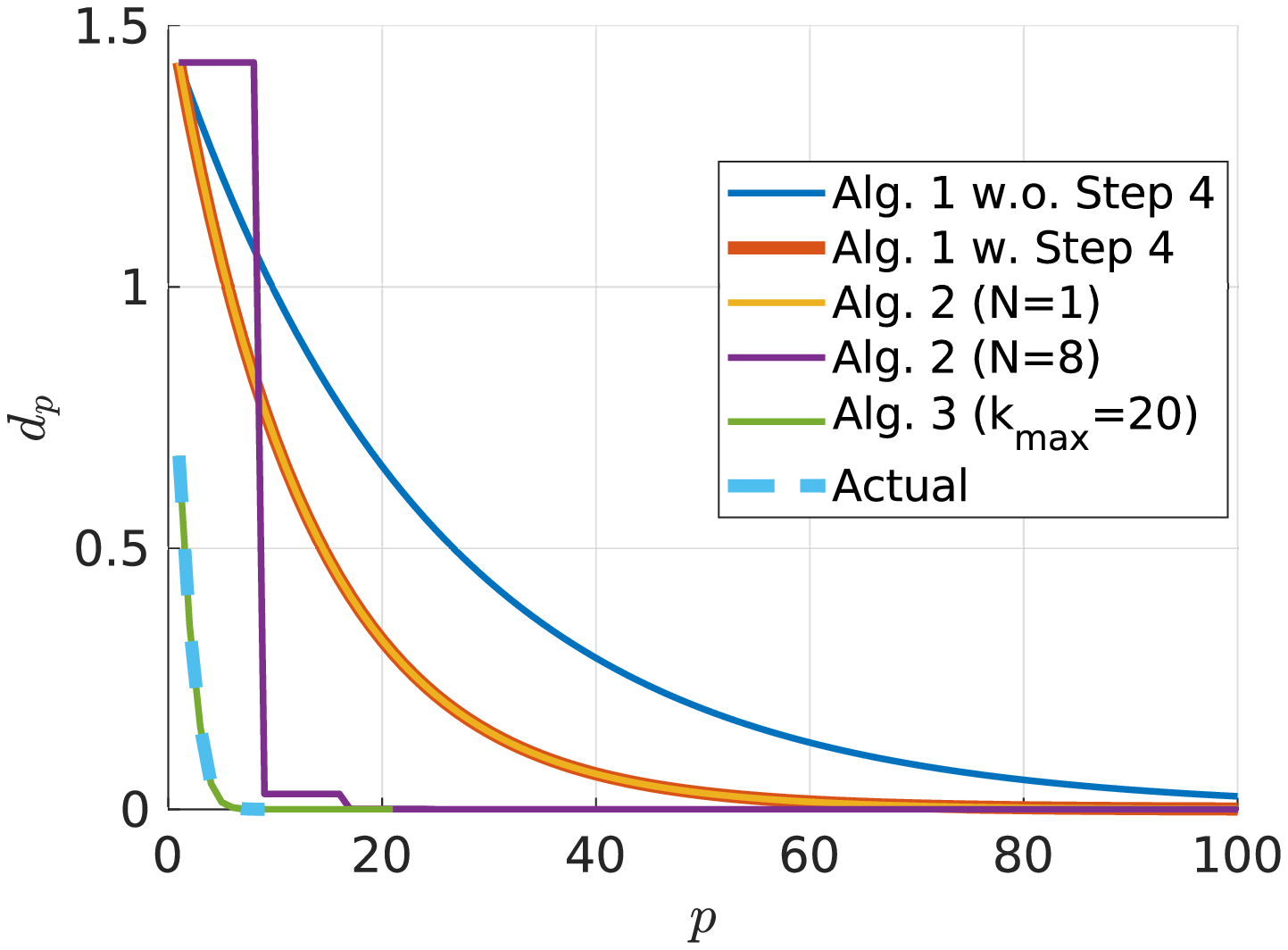}
		\caption{Two-dimensional example}
		\label{fig:ex2_curve}
     \end{subfigure}
     \begin{subfigure}[b]{0.32\textwidth}
		\includegraphics[width=\textwidth]{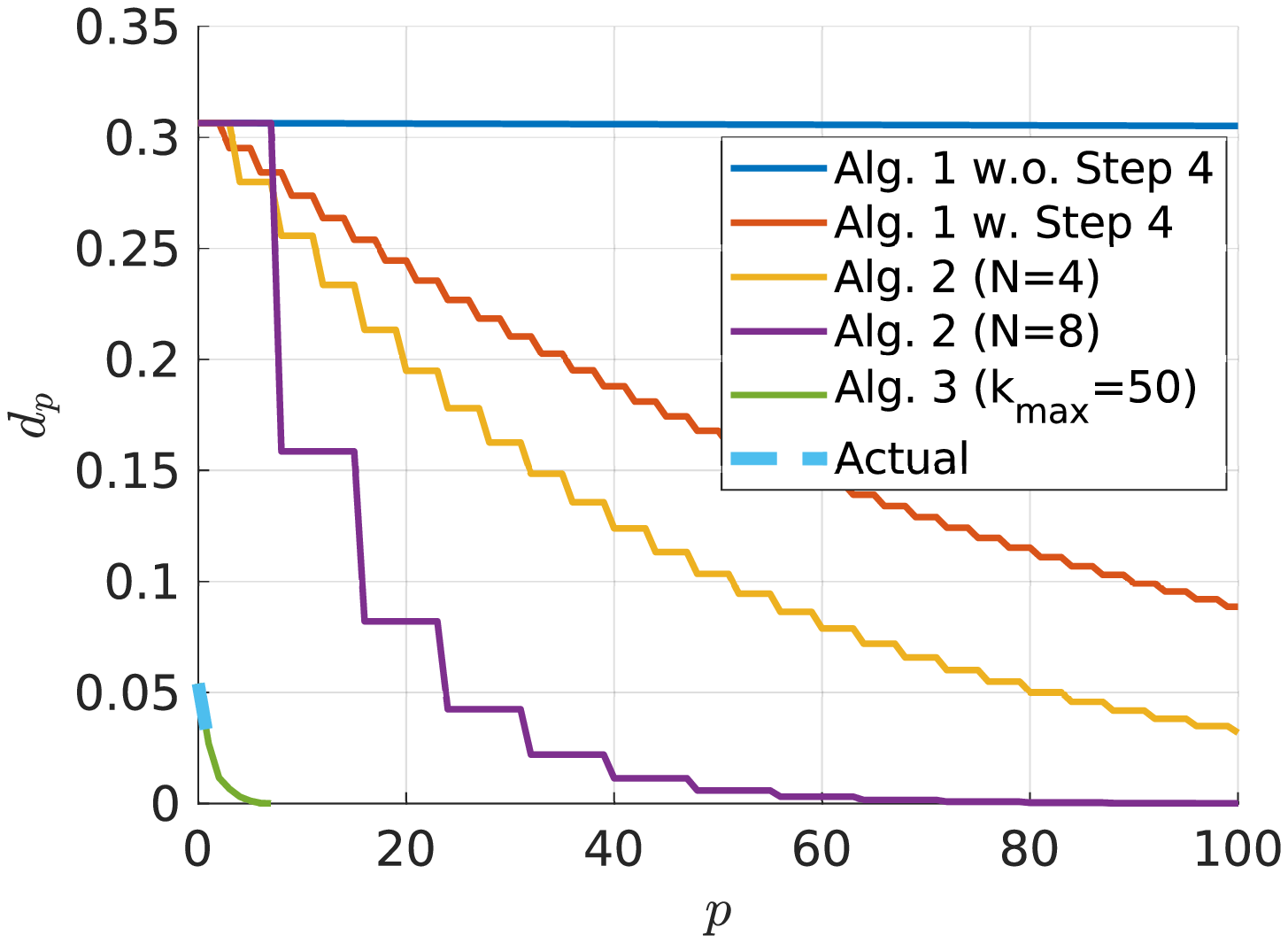}
		\caption{Lane-keeping example}
		\label{fig:lk_curve}
     \end{subfigure}
     \begin{subfigure}[b]{0.32\textwidth}
		\includegraphics[width=\textwidth]{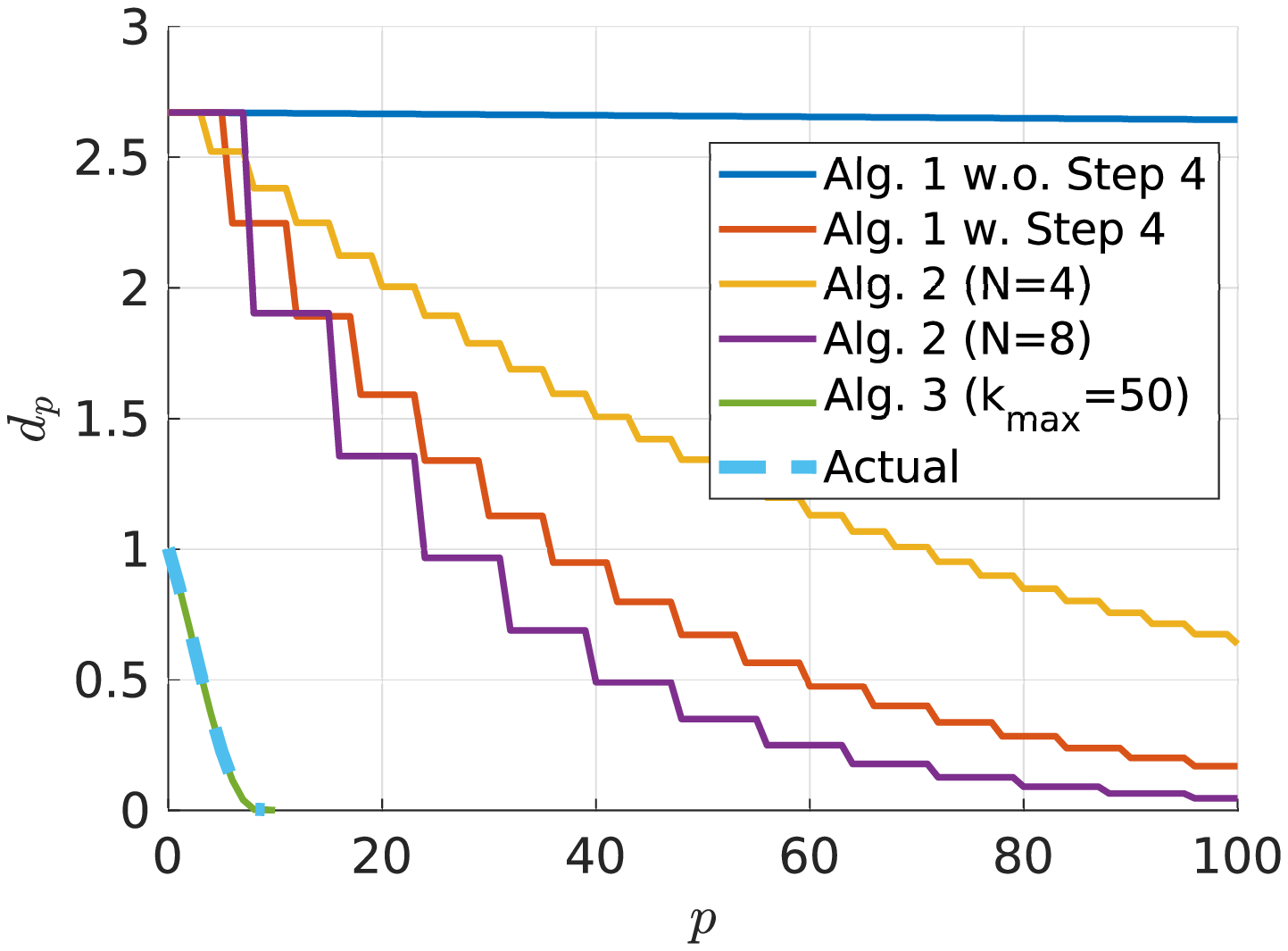}
		\caption{Biped example}
		\label{fig:biped_curve}
     \end{subfigure}
	\caption{The Hausdorff distance $d_{p}$ (light blue dash curve) in \eqref{eqn:dp}  and its upper bounds estimated by Alg. \ref{alg:1} (blue and red curves), Alg. \ref{alg:2} (yellow and purple curves), and Alg. \ref{alg:3} (green curve) versus the preview horizon $p$. }
	\label{fig:curves} 
\end{figure*}

\subsection{Lane-keeping control} \label{sec:lk} 
We consider the linearized bicycle model with respect to the longitudinal velocity $30m/s$ in \cite{smith2016interdependence} with the same parameters. The states $x=(y, v, \Delta \Psi, r)$ are the lateral displacement $y$, lateral velocity $v$, yaw  angle $\Delta \Psi$ and yaw rate $r$; the input is the steering angle $\delta_{f}$; the disturbance $r_{d}$ is the road curvature, with $ \vert r_{d} \vert \leq 0.05$. The safe set of the system is the set of state-input pairs satisfying $ \vert y \vert \leq 0.9$, $ \vert v \vert \leq 1.2$, $ \vert \Delta \Psi \vert \leq 0.05$, $ \vert r \vert \leq 0.3$ and $ \vert \delta_{f} \vert \leq \pi/2$. 

The preview on the road curvature $r_{d}$ is assumed to be available. We select $p_0=0$ and run Algs. \ref{alg:1}, \ref{alg:2} and  \ref{alg:3}. Alg. \ref{alg:3} obtains $ \overline{p} = 7$, implying that $\Proj{n}(C_{max,p})$ converges to $C_{max,co}$ at $p=7$. This observation suggests that the lane-keeping controller gains most from the first $7$-step preview. The upper bounds  on $d_{p}$ from the three algorithms are shown in Fig. \ref{fig:lk_curve}. The upper bound obtained by Alg. \ref{alg:3} is the tightest. The computation time of the three algorithms is listed in the second column of Table \ref{tab:time}.

\subsection{Biped walking pattern generation} \label{sec:biped} 
In \cite{kajita2003biped, wieber2006trajectory}, preview of the zero-moment point (ZMP) reference is used to control the center of mass (CoM) of bipedal robots. We consider the discrete-time lateral dynamics of the biped in \cite{wieber2006trajectory}, with the same parameters. The states $(x, \dot{x}, \ddot{x})$ are the lateral position $x$ of the CoM and its first and second derivatives; the input is the jerk $\dddot{x}$ of $x$. The target is to control the lateral position $x$ such that the ZMP $z = x + (h_{CoM}/g)\ddot{x}$ tracks a given reference closely, where $h_{CoM}$ and $g$ are the robot altitude and the acceleration due to gravity. The ZMP reference is typically a periodic square signal (see \cite{kajita2003biped, wieber2006trajectory}). Here we assume that the reference signal is a trajectory of the uncertain system
\begin{align} \label{eqn:sys_ref} 
   \overline{z} (t+1) = 0.15 \overline{z}(t) + d(t), 
\end{align}
with disturbance $ \vert d \vert \leq 0.085$. This system can generate periodic signals with maximal magnitude up to $0.1$. We couple the biped lateral dynamics and the reference dynamics in \eqref{eqn:sys_ref}, which leads to a four-dimensional system with state-input pairs satisfying $(x, \dot{x}, \ddot{x}, \overline{z})$. Based on the control target, we select the safe set $S_{xu}$ to be the set of state-input pairs satisfying $ \vert x+ (h_{CoM}/g)\ddot{x}-\overline{z} \vert \leq 0.1$, $ \vert x \vert \leq 0.1$, $ \vert \dot{x} \vert \leq 10$, $ \vert \ddot{x} \vert \leq 10$, $ \vert \overline{z} \vert \leq 0.1 $ and $ \vert \dddot{x}\vert \leq 100$.

The preview on the disturbance $d$ in \eqref{eqn:sys_ref} is available, induced from the preview of ZMP reference\cite{wieber2006trajectory}.  We select $p_0=0$. Alg. \ref{alg:3} returns $\overline{p}= 10$, which implies that $\Proj{n}(C_{max,p}) $ converges to $ C_{max,co}$ at $p=10$. This result suggests that the first $10$-step preview on the ZMP reference is most useful for maintaining the lateral position of the robot in the safe set. We also depict the upper bounds of $d_{p}$ obtained by Algs. \ref{alg:1}, \ref{alg:2} and \ref{alg:3} in Fig. \ref{fig:biped_curve}. The conservativeness of those upper bounds are similar to the 2D example in Section \ref{sec:ex2}.  The computation time of the three algorithms is listed in the last column of Table \ref{tab:time}.  

\subsection{Blade pitch control of wind turbine} \label{sec:wind_turbine} 
Preview of incoming wind events can be measured by Lidar scanners, whose usage has been widely studied in the literature of wind turbine control\cite{menezes2018review,sinner2021experimental}. In this subsection, we demonstrate the use of our methods in analyzing the feasible domain of the preview-based constrained MPC in \cite{sinner2021experimental} for blade pitch control of a wind turbine. We consider the discrete-time linear wind turbine dynamics in \cite{sinner2021experimental} with states $x=(\delta \Omega, \int \delta \Omega, \delta \beta)$, input $\Delta \beta$ and disturbance $\delta v$. The states $\delta \Omega$,  $\int \delta \Omega$ and $\delta \beta$ are the rotor speed ($rad/s$) relative to a nominal speed, the integral of $\delta \Omega$ and  the pitch angle  ($\deg$) relative to a nominal angle; the input $\Delta \beta$ is the increment of $\delta \beta$ in one step; the disturbance $\delta v$ is the wind speed relative to a nominal wind speed.  We use the same parameters in \cite{sinner2021experimental}. The safe set $S_{xu}$ is the set of state-input pairs satisfying the state-input constraints $J x+ E\Delta \beta \leq l$ of the MPC proposed in \cite{sinner2021experimental} (with $J$, $E$ and $l$ defined in \cite{sinner2021experimental}), and some large enough state bounds $ \vert \delta \Omega \vert \leq 5$, $ \vert \int \delta \Omega \vert \leq 100$, and $ -4.53 \leq \delta \beta \leq 10.47$ to ensure the compactness of $S_{xu}$. 

The disturbance $\delta v$  on the wind speed can be previewed\cite{sinner2021experimental}. We select $C$ to be the maximal RCIS of the system without preview. Alg. \ref{alg:3} terminates with $\overline{p} = 4$. That is, $\Proj{n}(\mathcal{F}_{p}(C))$ converges to $C_{max,co}$ at $p=4$, which suggests that the MPC benefits most from the first $4$-step preview in terms of feasible domain size.
We also run Algs. \ref{alg:1}  and \ref{alg:2} as in the previous examples. But since Alg. \ref{alg:3} terminates quickly and provides the tightest bound, the  results from the other algorithms are omitted.

\section{Conclusion} \label{sec:conclude} 
In this work, we study the impact of different preview horizons on the safety of discrete-time systems. For general nonlinear system, we derive a novel outer bound on the maximal RCIS $C_{max,p}$. For linear systems, under mild conditions, we prove that the safety regret $d_{p}$ decays to zero exponentially fast, which indicates that the marginal value of the preview information decays to zero exponentially fast with the preview horizon. We further develop algorithms to compute upper bounds on the safety regret $d_{p}$. 
We also adapt the established theoretical and algorithmic tools to analyze the convergence of the feasible domain of a preview-based robust MPC.
The efficiency of the proposed algorithms are verified with four examples. 

In this work, we assume that we have accurate preview of future disturbances for a finite time horizon. For future work, we want to relax this assumption and study the case where the preview is uncertain (that is, the measurements of future disturbances are noisy). This is equivalent to having imperfect state information on the preview part of the states of the augmented system $\Sigma_{p}$ and safety control with imperfect state information is in general hard \cite{yang2020efficient,artstein2011set}. Besides, our current work deals with a single disturbance with a fixed preview horizon. As part of the future work, we want to extend our results to systems with multiple disturbances, where each disturbance may have a different preview horizons. 

\appendix
\subsection{Proof of Theorem \ref{thm:CIS_delay}}
\begin{proof}
First, we want to show the set $C_{max,p,co}$ in \eqref{eqn:C_max_p_co} is a CIS of the system $\mathcal{D}(\Sigma_{p})$ within the safe set  $S_{xu,p}\times D$. Let $(x(0), d_{1:p}(0))$ be any point in $ C_{max,p,co}$. By \eqref{eqn:C_max_p_co} , there exists $u(0)$, ..., $u(p-1)$ such that  $(x(t),u(t))$ is in $S_{xu}$ for $t$ from  $0$ to $p-1$ and $x(p)$ is in $ C_{max,co}$. Since $C_{max,co}$ is controlled invariant for $\mathcal{D}(\Sigma)$, there exist $u(p)$ and $u_d(0)\in D$ such that $(x(p), u(p))$ is in $ S_{xu}$ and $x(p+1) = f(x(p),u(p),u_d(0))$ is in $ C_{max,co}$. Thus, for control inputs $u(0)$ and $u_d(0)$, {it can be checked} that  $(x(0), d_{1:p}(0), u(0), u_d(0))\in  S_{xu,p}\times D$ and the next state $$(x(1), d_{2:p}(0),u_d(0))\in C_{max,p,co},$$ where $x(1)=f(x(0),u(0),d_1(0))$.  Thus, $C_{max,p,co}$ is a CIS of $\mathcal{D}(\Sigma_{p})$ within $S_{xu,p,co}$.

{
Next, denote the maximal CIS of $\mathcal{D}(\Sigma_{p})$ in $S_{xu,p}\times D$ by $\overline{C}_{max,p,co}$. It can be checked that $\Proj{n}(\overline{C}_{max,p,co})$ is a CIS of $\mathcal{D}(\Sigma)$ in $S_{xu}\times D$, and thus is contained in $C_{max,co}$. Thus,
\begin{align}\label{eqn:C_max_co_contain}
\overline{C}_{max,p,co} \subseteq C_{max,co}\times D^{p},
\end{align}
which proves \eqref{eqn:delay_relation}.
Since $\overline{C}_{max,p,co}$ is the maximal CIS of $\mathcal{D}(\Sigma_{p})$ in $S_{xu,p}\times D$, $\overline{C}_{max,p,co} = Pre_{\mathcal{D}(\Sigma_{p})}^{p}(C_{max,p,co})$ and thus by \eqref{eqn:C_max_co_contain},
\begin{align}
\overline{C}_{max,p,co} \subseteq Pre_{\mathcal{D}(\Sigma_{p})}^{p}(C_{max,co}\times D^{p}) = C_{max,p,co}.
\end{align}
Therefore, we prove that $\overline{C}_{max,p,co} = C_{max,p,co}$.}
\end{proof}

\subsection{Proof of Lemma \ref{lem:forced_eq}}
\begin{proof} 
We first show that $\Proj{n}(C_{max,p})$ is convex and compact. Note that
Assumption \ref{asp:cmpt} implies that the safe set $S_{xu,p}$ is convex and compact. Thus, the maximal RCIS $C_{max,p}$ is compact. Also, for linear systems, the convex hull of any RCIS is also a RCIS. Thus, the maximal RCIS of $\Sigma_{p}$ in $S_{xu,p}$ must be convex (otherwise we can take the convex hull of the maximal RCIS and obtain a larger RCIS). Since $C_{max,p}$ is convex and compact, the projection $\Proj{n}(C_{max,p})$ is convex and compact.

Next, we show that $\Proj{n}(C_{max,p})$ is a CIS of $\mathcal{D}(\Sigma)$ within $S_{xu}\times D$ by definition. Let $x $ be an arbitrary point in $\Proj{n}(C_{max,p})$. There exists $d_{1:p}\in D^{p}$ such that $(x,d_{1:p})\in C_{max,p}$. As $C_{max,p}$ is an RCIS of $\Sigma_{p}$, there exists $u_0$ such that $(x,u_0)\in S_{xu}$ such that $(Ax+Bu_0+Ed_1,d_{2:p},d)\in C_{max,p}$ for any $d\in D$. Thus, there exists $u$ and $u_{d}$ such that $(x,u_,u_{d})\in S_{xu}\times D$ and $Ax+Bu+Eu_{d}\in \Proj{n}(C_{max,p})$ (one feasible choice is $u=u_0$, $u_{d}=d_1$). Thus, by definition, $\Proj{n}(C_{max,p})$ is a CIS of $\mathcal{D}(\Sigma)$ in $S_{xu}\times D$.

Finally, \cite[Proof of Theorem 3.3]{anevlavis2019computing} shows\footnote{\cite[Proof of Theorem 3.3]{anevlavis2019computing} only proves the case without input constraints, which can be easily extended to the case with state-input constraints according to \cite[Remark 1]{anevlavis2019computing}. Also, though \cite{anevlavis2019computing} mainly considers controllable systems, the part used to prove (ii) holds for any linear system.} that for any nonempty convex compact CIS $C$ of a linear system $\Sigma$ in safe set $S_{xu}$, there always exists a forced equilibrium $(x_{e},u_{e})\in S_{xu}$ of the system such that $x_{e}\in C$. Thus, (i) implies (ii). 
\end{proof}

\subsection{Proof of Lemma \ref{lem:pre_expand}}
\begin{proof}
We first show that  the $1$-step backward reachable set of the projection $\Proj{n}(C_{max,p})$ is contained in the projection $\Proj{n}(C_{max,p+1})$, that is 
\begin{align} \label{eqn:1_step_contain} 
&Pre_{\mathcal{D}(\Sigma)}(\Proj{n}(C_{max,p}), S_{xu}\times D)
\subseteq \Proj{n}(C_{max,p+1}).    
\end{align}

By Theorem \ref{thm:bnds}, $C_{max,p}\times D \subseteq C_{max,p+1}$, which implies 
\begin{align}
	&Pre_{\Sigma_{p+1}}(C_{max,p}\times D, S_{xu,p+1}) \nonumber\\
	&\subseteq Pre_{\Sigma_{p+1}}(C_{max,p+1},S_{xu,p+1})=  C_{max,p+1}.
\end{align}
Thus, to prove \eqref{eqn:1_step_contain}, it is sufficient to show 
\begin{align}
	&Pre_{\mathcal{D}(\Sigma)}(\Proj{n}(C_{max,p}), S_{xu}\times D) \nonumber\\
	&\subseteq \Proj{n}(Pre_{\Sigma_{p+1}}(C_{max,p}\times D, S_{xu,p+1})).
\end{align}
	We select an arbitrary point $x_0 $ such that $$x_0\in Pre_{\mathcal{D}(\Sigma)}(\Proj{n}(C_{max,p}),S_{xu}\times D).$$ By the definition of $Pre$, there exists $u_0$, $d_0$ such that the point $(x_0,u_0,d_0)$ is in $ S_{xu}\times D$ and the point ${x_1=f(x_0,u_0,d_0)}$ is in $ \Proj{n}(C_{max,p})$. Also, there exists $d_{1:p}\in D^{p}$ such that $(x_1, d_{1;p})\in C_{max,p}$. That is, $(f(x_0,u_0,d_0), d_{1:p},D) \subseteq C_{max,p}\times D$. Thus, $(x_0, d_0, d_{1:p})\in Pre_{\Sigma_{p+1}}(C_{max,p}\times D)$, which implies $x_0\in \Proj{n}(Pre_{\Sigma_{p+1}}(C_{max,p}\times D)).$ Thus, \eqref{eqn:1_step_contain} is proven.

Next, by taking $Pre$ on both sides of \eqref{eqn:1_step_contain}, we have
\begin{align} \label{eqn:2_step_contain} 
&Pre^{2}_{\mathcal{D}(\Sigma)}(\Proj{n}(C_{max,p}), S_{xu}\times D) \nonumber\\
&\subseteq Pre_{\mathcal{D}(\Sigma)}(\Proj{n}(C_{max,p+1}), S_{xu}\times D)\nonumber\\
& \subseteq  \Proj{n}(C_{max,p+2}),
\end{align}
where the second inclusion is due to \eqref{eqn:1_step_contain}. Following the pattern in \eqref{eqn:1_step_contain} and \ref{eqn:2_step_contain}, \eqref{eqn:k_step_contain} follows by induction. 
\end{proof}

\subsection{Proof of Theorem \ref{thm:conv_suff}}
{Recall that by Remark \ref{rem:origin}, Assumptions \ref{asp:int} and \ref{asp:int2}, the sets $S_{xu}\times D$ and $\Proj{n}(C_{max,p_0})$ all contain the origin in their interioirs.}
To prove Theorem \ref{thm:conv_suff}, we first need the following lemma. 
\begin{lemma} \label{lem:one_step_conv} 
Under the same conditions of Theorem \ref{thm:conv_suff}, for any $ \xi\in (0,1]$, the $N$-step backward reachable set of $ \xi C_{max,co}$ satisfies that for $ \xi \leq  \lambda \gamma$, 
\begin{align} \label{eqn:one_step_conv_1} 
	Pre_{\mathcal{D}(\Sigma)}^{N}( \xi C_{max,co}, S_{xu}\times D) \supseteq\frac{\xi}{ \lambda} C_{max,co}; 
\end{align}
for $\xi > \lambda\gamma$,
\begin{align}
	&Pre_{\mathcal{D}(\Sigma)}^{N}( \xi C_{max,co}, S_{xu}\times D)  \supseteq g(\xi) C_{max,co}, \label{eqn:one_step_conv_2} 
\end{align}
where $g(\xi)= \xi(1-\gamma)/(1-\gamma \lambda) + \gamma(1-\lambda)/(1-\gamma \lambda)$.
\end{lemma}

\begin{proof}
First, it can be checked that for any $ \gamma' \leq  \gamma$, $\gamma' C_{max,co}$ is an $N$-step $ \lambda$-contractive CIS of $\mathcal{D}(\Sigma)$ in $\gamma'/\gamma (S_{xu}\times D)$, which is a subset of $S_{xu}\times D$. 

Case 1: $ \xi \leq  \lambda \gamma$. Since $\xi/ \lambda \leq  \gamma$, $ (\xi/ \lambda) C_{max,co}$ is an $N$-step $ \lambda$-contractive CIS of $\mathcal{D}(\Sigma)$ in $S_{xu}\times D$, which implies \eqref{eqn:one_step_conv_1}. 

Case 2: $\xi > \lambda \gamma$. We first separate $\xi $ into two parts, that is
\begin{align} \label{eqn:C_co_decomp} 
    \xi = \xi_1 +\lambda \xi_2 
\end{align}
where   $\xi_1=  {(\xi - \lambda \gamma)}/{(1- \lambda \gamma)}$, $\xi_2= \gamma( {(1 - \xi)}/{(1- \lambda \gamma)} )$.

Since $C_{max,co}$ is convex, we can separate $\xi C_{max,co}$ into two parts, that is
  $ \xi C_{max,co} = \xi_1 C_{max,co} + \lambda \xi_2 C_{max,co}.$
It can be shown that for any convex set $C$, and any $a$, $b\in [0,1]$, 
{
	\begin{align} \label{eqn:pre_ineq} 
	&Pre_{\mathcal{D}(\Sigma)}^{N}(C, S_{xu}\times D) \supseteq Pre_{\mathcal{D}(\Sigma)}^{N}(aC, b(S_{xu}\times D))+\nonumber \\
	& \ \ \ \ Pre_{\mathcal{D}(\Sigma)}^{N}((1-a)C, (1-b)(S_{xu}\times D)).
\end{align}
By \eqref{eqn:C_co_decomp} and \eqref{eqn:pre_ineq}, we have
\begin{align} \label{eqn:lemma_3_0} 
 &Pre_{\mathcal{D}(\Sigma)}^{N}(\xi C_{max,co}, S_{xu}\times D) \supseteq \nonumber\\
 &\ \ Pre_{\mathcal{D}(\Sigma)}^{N}( \xi_1 C_{max,co}, \xi_1 (S_{xu}\times D)) +\nonumber\\
 &\ \ \ \ Pre_{\mathcal{D}(\Sigma)}^{N}(  \lambda\xi_2 C_{max,co}, (1- \xi_1)(S_{xu}\times D)).
\end{align}
}
Since $ \xi_1 C_{max,co}$ is a CIS of $\mathcal{D}(\Sigma)$ in $ \xi_1 (S_{xu}\times D)$, we have
\begin{align} \label{eqn:lemma_3_1}
    Pre_{\mathcal{D}(\Sigma)}^{N}( \xi_1 C_{max,co}, \xi_1 (S_{xu}\times D)) \supseteq \xi_1 C_{max,co}.
\end{align}
Note that $1- \xi_1 = \xi_2/ \gamma$. Also, for a linear $\mathcal{D}(\Sigma)$, we have that for any set $C$ and $a>0$, 
\begin{align} \label{eqn:pre_linear} 
  Pre_{\mathcal{D}(\Sigma)}^{N}(C, S_{xu}\times D) = \frac{1}{a}  Pre_{\mathcal{D}(\Sigma)}^{N}(aC, a(S_{xu}\times D)).
\end{align}
Thus, by \eqref{eqn:pre_linear},  
\begin{align} \label{eqn:lemma_3_2} 
   &Pre_{\mathcal{D}(\Sigma)}^{N}(  \lambda\xi_2 C_{max,co}, (1- \xi_1)(S_{xu}\times D))\nonumber\\ 
   =&Pre_{\mathcal{D}(\Sigma)}^{N}(  \lambda\xi_2 C_{max,co}, \xi_2/ \gamma(S_{xu}\times D))\nonumber\\
   =& (\xi_2/ \gamma) Pre_{\mathcal{D}(\Sigma)}^{N}(  \lambda \gamma C_{max,co}, S_{xu}\times D)\nonumber\\
   \supseteq & ( \xi_2/ \gamma) \gamma C_{max,co} = \xi_2 C_{max,co}. 
\end{align}
By \eqref{eqn:lemma_3_0}, \eqref{eqn:lemma_3_1} and \eqref{eqn:lemma_3_2}, we have 
\begin{align}
   Pre_{\mathcal{D}(\Sigma)}^{N}(\xi C_{max,co}, S_{xu}\times D) \supseteq 
   ( \xi_1 + \xi_2) C_{max,co}
\end{align}
{One can see that} $ \xi_1 + \xi_2 = 1- (1-\xi){(1- \gamma)}/{ (1- \gamma \lambda)}, $
which proves \eqref{eqn:one_step_conv_2}. 
\end{proof}

Now we are ready to prove Theorem \ref{thm:conv_suff}.

\emph{Proof of Theorem \ref{thm:conv_suff}:} We only prove the case of $ \lambda_0 \leq \lambda \gamma$, since the case of $\lambda_0 > \lambda \gamma$ can be easily extended from this proof. Given $ \lambda_0 \leq  \lambda \gamma$, by \eqref{eqn:one_step_conv_1}, we have 
\begin{align} \label{eqn:thm_4_1} 
	Pre_{\mathcal{D}(\Sigma)}^{N}( \lambda_0 C_{max,co}, S_{xu}\times D) \supseteq\frac{\lambda_0}{ \lambda} C_{max,co}.
\end{align}
If $ \lambda_0/ \lambda \leq  \lambda\gamma$, we can compute the $N$-step backward reachable sets of both sides of \eqref{eqn:thm_4_1} and apply \eqref{eqn:one_step_conv_1} again, which leads to
\begin{align} \label{eqn:thm_4_1_2} 
	Pre_{\mathcal{D}(\Sigma)}^{2N}( \lambda_0 C_{max,co}, S_{xu}\times D)	&\supseteq Pre_{\mathcal{D}(\Sigma)}^{N}(\frac{\lambda_0}{ \lambda} C_{max,co})\nonumber\\
	&\supseteq \frac{\lambda_0}{ \lambda^{2}} C_{max,co}.
\end{align}
Following the pattern of \eqref{eqn:thm_4_1} and \eqref{eqn:thm_4_1_2}, as long as $ \lambda_0/ \lambda^{k-1} \leq  \lambda \gamma$, we can prove by induction that 
\begin{align} \label{eqn:thm_4_2} 
	Pre_{\mathcal{D}(\Sigma)}^{kN}( \lambda_0 C_{max,co}, S_{xu}\times D)\supseteq \frac{\lambda_0}{ \lambda^{k}} C_{max,co}.
\end{align}
We denote the minimal $k$ such that $ \lambda_0/ \lambda^{k}> \lambda \gamma$ by $k_0$. {The expression of $k_0$ is}
\begin{align} \label{eqn:thm_4_3}
    k_0 = \bigg\lceil  \frac{ \log \lambda_0 - \log \gamma}{ \log \lambda} -1 \bigg\rceil.
\end{align}
By \eqref{eqn:thm_4_2} and \eqref{eqn:thm_4_3}, we have proven \eqref{eqn:C_max_conv_1} for $k\leq  k_0$. 

Next, since $ \lambda_0/ \lambda^{k_0} > \lambda \gamma$, by taking $N$-step backward reachable sets on both sides of \eqref{eqn:thm_4_2} with $k=k_0$ and applying \eqref{eqn:one_step_conv_2}, we have
\begin{align} \label{eqn:thm_4_5} 
	&Pre_{\mathcal{D}(\Sigma)}^{(k_0+1)N}( \lambda_0 C_{max,co}, S_{xu}\times D)\nonumber\\
	&\supseteq Pre_{\mathcal{D}(\Sigma)}^{N}\left(\frac{\lambda_0}{ \lambda^{k_0}} C_{max,co}, S_{xu}\times D\right)\nonumber\\
	& \supseteq g( \lambda_0/ \lambda^{k_0}) C_{max,co},
\end{align}
where the function $g(\xi)$ is defined in Lemma \ref{lem:one_step_conv}. We define recursively
$g^{1}(\xi) =  g(\xi)$, and $
g^{k}(\xi) =  g(g^{k-1}(\xi))$.
It can be checked that if $\xi > \lambda \gamma$, then $g^{k}(\xi) > \lambda \gamma$ for all $k\geq 1$. Thus, $g^{k}( \lambda_0/ \lambda^{k_0}) > \lambda \gamma$ for all $k\geq 1$. Following the pattern of \eqref{eqn:thm_4_5}, we can prove by induction that 
\begin{align} \label{eqn:thm_4_6}  
	&Pre_{\mathcal{D}(\Sigma)}^{(k_0+k)N}( \lambda_0 C_{max,co}, S_{xu}\times D) 
	\supseteq  g^k( \lambda_0/ \lambda^{k_0}) C_{max,co}
\end{align}
Finally, it can be checked that 
\begin{align} \label{eqn:thm_4_7}
    g^{k}(\xi) = 1- (1-\xi)\left( \frac{1- \gamma}{1- \gamma \lambda}\right)^{k-k_0} .
\end{align}
By \eqref{eqn:thm_4_6} and \eqref{eqn:thm_4_7}, we have proven \eqref{eqn:C_max_conv_2} for $k\geq k_0$. 

The proof for the case of $ \lambda_0 > \lambda \gamma$ follows exactly the same arguments from \eqref{eqn:thm_4_5} to \eqref{eqn:thm_4_7}, with $k_0=0$.\quad \quad \quad \quad\quad\quad\quad\ \QED

\subsection{Proof of Lemma \ref{lem:contraction} }
\begin{proof}
By the proof of Proposition 23 of \cite{de2004computation}, since the system $\mathcal{D}(\Sigma)$ is stabilizable and the safe set $S_{xu}\times D$ contains the origin in the interior, there exists a $\lambda_{a}$-contractive ellipsoidal CIS $ \mathcal{E}$ within the safe set $S_{xu}\times D$ for some $ \lambda_{a} \in [0,1)$. Clearly, $ \mathcal{E} \subseteq C_{max,co}$, and thereby $C_{max,co}$ contains the origin in the interior.  

Let $ \lambda_{in}$ be the maximal $ \lambda$ such that $ \lambda C_{max,co} \subseteq \mathcal{E}$. As $\mathcal{E}$ contains the origin in the interior, $ \lambda_{in} > 0$. Since $ \mathcal{E}$ is $ \lambda$-contractive, we have 
\begin{align}
\lambda_{in} C_{max,co}	 \subseteq \mathcal{E} \subseteq Pre^{k}_{ \mathcal{D}(\Sigma)}( \lambda_{a}^{k}\mathcal{E},S_{xu}\times D).
\end{align}
Let $ \lambda$ be an arbitrary number in $(0,1)$. Since $ \lambda_{a} \in [0,1)$ and $\mathcal{E} \subseteq C_{max,co}$, there exists $N>0$ such that $ \lambda_{a}^{N} \mathcal{E} \subseteq \lambda \lambda_{in} C_{max,co} $. Thus,
\begin{align}
	\lambda_{in} C_{max,co} &\subseteq   Pre_{\mathcal{D}(\Sigma)}^{N}( \lambda_{a}^{N} \mathcal{E}, S_{xu}\times D)\\
							&\subseteq Pre_{\mathcal{D}(\Sigma)}^{N}( \lambda \lambda_{in} C_{max,co}, S_{xu}\times D).
\end{align}
Thus, $ \lambda_{in} C_{max,co}$  is an $N$-step $ \lambda$-contractive CIS of  $ \mathcal{D}(\Sigma)$ in $S_{xu}\times D$.
\end{proof}

\subsection{Proof of Theorem \ref{thm:parameters_asym}}
\begin{proof}
By the definition of $c_0$, $\mathcal{E}(c_0)$ is a $ \lambda_{a}$-contractive CIS of $\mathcal{D}(\Sigma)$ in $S_{xu}\times D$. Therefore, $\mathcal{E}(c_0)$ is a subset of $C_{max,co}$. Also, $ \gamma C_{max,co} \subseteq \mathcal{E}( \gamma c_{out})= \mathcal{E}(c_0)$. Thus, we have
\begin{align}
    \gamma C_{max,co} \subseteq \mathcal{E}(c_0) \subseteq C_{max,co}.
\end{align}
Since $ \mathcal{E}(c_0)$ is $\lambda_{a}$-contractive, we have for any $k  \geq 1$
\begin{align} \label{eqn:thm_6_1} 
   Pre_{\mathcal{D}(\Sigma)}^{k}( \lambda_{a}^{k}\mathcal{E}(c_0), S_{xu}\times D) \supseteq \mathcal{E}(c_0) \supseteq \gamma C_{max,co}.
\end{align}
Since $C_{max,co} \supseteq \mathcal{E}(c_0)$, 
\begin{align} \label{eqn:thm_6_2} 
   Pre_{\mathcal{D}(\Sigma)}^{k}( \lambda_{a}^{k}C_{max,co},S_{xu}\times D) \nonumber\\
   \supseteq  Pre_{\mathcal{D}(\Sigma)}^{k}( \lambda_{a}^{k} \mathcal{E}(c_0),S_{xu}\times D) 
\end{align}
By \eqref{eqn:thm_6_1} and \eqref{eqn:thm_6_2},
\begin{align}
   Pre_{\mathcal{D}(\Sigma)}^{k}( \lambda_{a}^{k}C_{max,co},S_{xu}\times D) \supseteq \gamma C_{max,co}.
\end{align}
By definition of $N$, $\lambda_{a}^{N} <\gamma$. Thus, $ \lambda  = \lambda_{a}^{N}/ \gamma<1$.
\begin{align}
   &Pre_{\mathcal{D}(\Sigma)}^{N}(  \lambda \gamma C_{max,co},S_{xu}\times D)\nonumber \\
	=&Pre_{\mathcal{D}(\Sigma)}^{N}( \lambda_{a}^{N}C_{max,co},S_{xu}\times D) \supseteq \gamma C_{max,co}.
\end{align}
That is, $ \gamma C_{max,co}$ is $N$-step $\lambda$-contractive.
\end{proof}

\subsection{Proof of Lemma \ref{lem:contraction_0}}
\begin{proof}
To prove the lemma, it suffices to show that $Pre_{\mathcal{D}(\Sigma)}^{n}(0,S_{xu}\times D)$ contains the origin in the interior.   Since the system is controllable, there exists a feedback gain  $K = [K_1\ K_2]$ such that the closed-loop system matrix $A_{c} = A+K_1B+ K_2E$ has all eigenvalues equal to zero. That implies $A_{c}^{n}=0$. In other words, under the control of $u=K_1x$, $u_{d}=K_2x$, the closed-loop system reaches the origin $0$ at step $n$ for all initial states $x_0\in \R^{n}$. Now let us consider the unit ball $\mathcal{B}$ in $\R^{n}$. Since $S_{xu}\times D$ contains the origin in the interior, there exists a positive scalar $\xi >0$ such that  $ \xi (A_{c}^{k}\mathcal{B}\times K A_{c}^{k}\mathcal{B}) $ is contained in $S_{xu}\times D$ for $k$ from $0$ to $n-1$. That implies that $ \xi \mathcal{B} $ is a subset of $Pre_{\mathcal{D}(\Sigma)}^{n}(\{0\}, S_{xu}\times D )$. Thus, $Pre_{\mathcal{D}(\Sigma)}^{n}(0,S_{xu}\times D)$ contains the origin in the interior. 

Since $Pre_{\mathcal{D}(\Sigma)}^{n}(0,S_{xu}\times D)$ is a CIS of  $\mathcal{D}(\Sigma)$ in $S_{xu}\times D$, the maximal CIS $C_{max,co}$ must contain $Pre_{\mathcal{D}(\Sigma)}^{n}(0,S_{xu}\times D)$ and thus contain the origin in the interior. Thus, there exists a positive scalar $ \gamma$ such that $ \gamma C_{max,co}$ is contained in $Pre_{\mathcal{D}(\Sigma)}^{n}(0,S_{xu}\times D)$.
\end{proof}

\bibliography{ref}{}
\bibliographystyle{plain}

\begin{IEEEbiography}[{\includegraphics[width=1in,height=1.25in,clip,keepaspectratio]{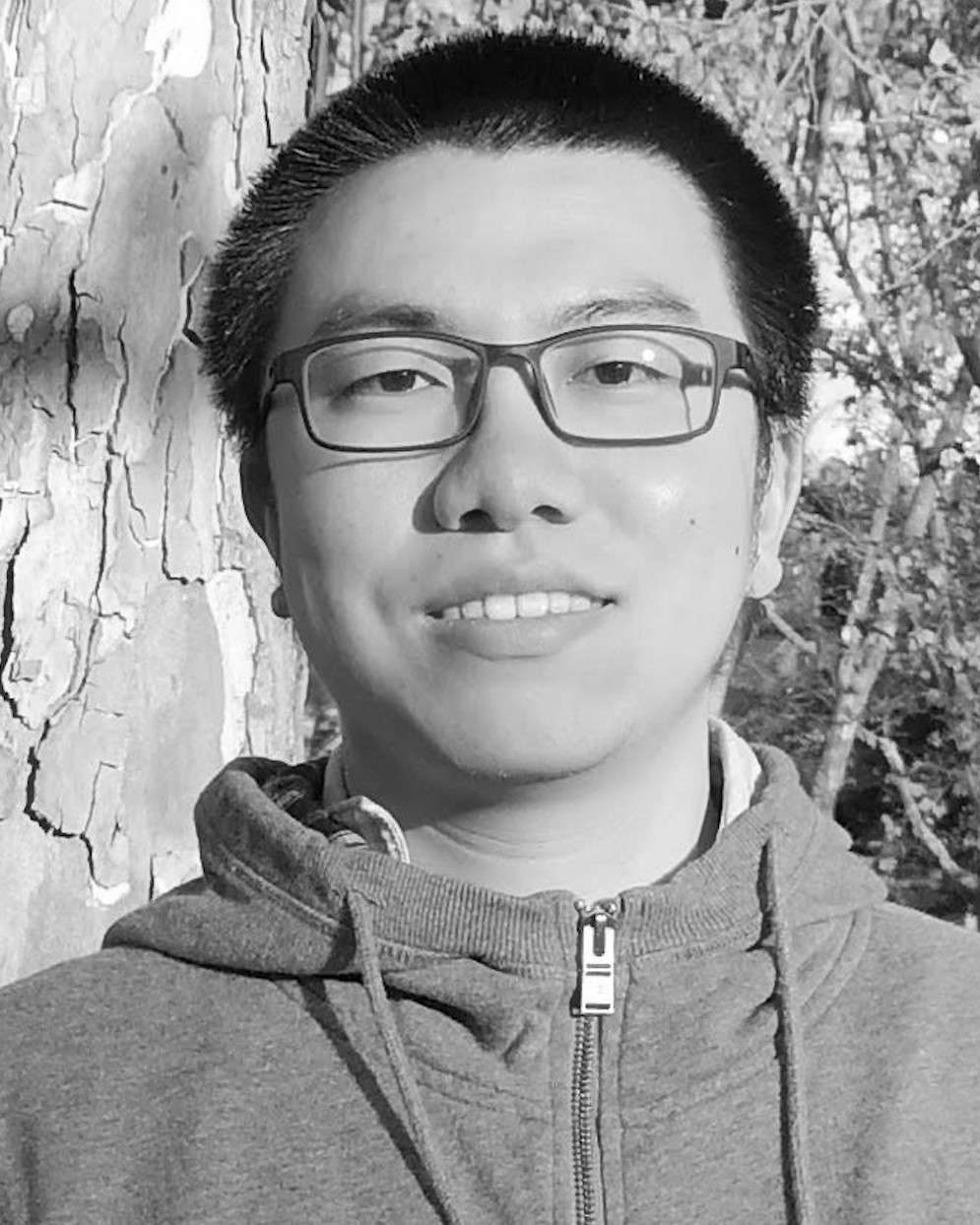}}]
	{Zexiang Liu} (Student Member, IEEE) received the B.S. degree in Engineering from Shanghai Jiao Tong University, Shanghai, China, in 2016, and the M.S. degree in Electrical and Computer Engineering from University of Michigan, Ann Arbor, MI, USA, in 2018. He is currently pursuing the Ph.D. degree in Electrical and Computer Engineering at the University of Michigan, Ann Arbor, MI, USA. 
	
	His current research interests lie in formal synthesis and verification for safety-critical systems, safe autonomy and system identification.
\end{IEEEbiography}

\begin{IEEEbiography}[{\includegraphics[width=1in,height=1.25in,clip,keepaspectratio]{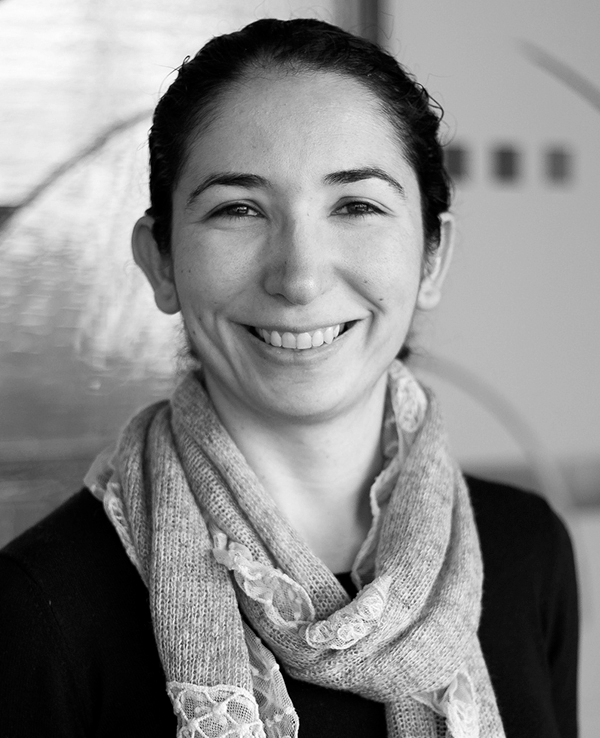}}]
	{Necmiye Ozay} (Senior Member, IEEE) received the B.S. degree from Bogazici University, Istanbul in 2004, the M.S. degree from the Pennsylvania State University, University Park in 2006 and the Ph.D. degree from Northeastern University, Boston in 2010, all in electrical engineering. Between 2010 and 2013, she was a postdoctoral scholar at California Institute of Technology, Pasadena. She joined the University of Michigan, Ann Arbor in 2013, where she is currently an associate professor of Electrical Engineering and Computer Science and Robotics. Her research interests include dynamical
	systems, control, optimization, formal methods with applications in cyber-physical systems, system identification, verification and validation, and safe autonomy.
\end{IEEEbiography}

\end{document}